\documentclass[reqno,11pt]{amsart}
\usepackage{graphicx}
\usepackage{enumitem}
\usepackage{todonotes}
\usepackage[numbers,sort&compress]{natbib}
\usepackage{hyperref}
\usepackage[foot]{amsaddr}
\usepackage[final,color,notref,notcite]{showkeys}
\usepackage{tikz-cd}
\usepackage{booktabs}
\usepackage{dsfont}
\usepackage{xcolor}
\usepackage{slashed}
\usepackage{colortbl}
\usepackage{algorithm}
\usepackage{algpseudocode}

\algnewcommand{\LineComment}[1]{\State \(\triangleright\) #1}

\usepackage[charter]{mathdesign}
\usepackage[scaled=.96,lf]{XCharter}

\theoremstyle{plain}
\newtheorem{thm}{Theorem}

\newtheorem{cor}[thm]{Corollary}
\newtheorem{lem}[thm]{Lemma}
\newtheorem{prop}[thm]{Proposition}

\theoremstyle{definition}
\newtheorem{defi}{Definition}
\newtheorem*{exa}{Example}

\theoremstyle{remark}
\newtheorem*{rem}{\bf Remark}

\definecolor{labelkey}{rgb}{0,.56,.7}

\let\bs\boldsymbol
\let\ox\otimes
\let\op\oplus
\let\ol\overline
\newcommand{\sox}{\slashed{\ox}}
\mathchardef\myhyph="2D

\DeclareMathAlphabet{\pazocal}{OMS}{zplm}{m}{n}   
\newcommand{\Pcal}{\pazocal{P}}
\newcommand{\Mcal}{\pazocal{M}}

\newcommand{\Ncal}{\pazocal{N}}

\newcommand{\Gcal}{\pazocal{G}}
\newcommand{\Fcal}{\pazocal{F}}

\def\bbC{\mathbb{C}}
\def\bbZ{\mathbb{Z}}
\def\bbR{\mathbb{R}}
\def\bbN{\mathbb{N}}
\def\bbI{\mathbb{I}}
\def\bbJ{\mathbb{J}}
\def\bbE{\mathbb{E}}

\newcommand{\dif}[1]{\mathrm{\,d} #1}             

\newcommand{\nn}{\nonumber}

\def\df{\overset{\mathrm{df}}{=}}
\newcommand{\ket}[1]{\mathop{|#1\rangle}\nolimits}

\def\ran{\rangle}
\def\lan{\langle}

\DeclareMathOperator*{\orb}{orb}

\newcommand{\haf}{\mathop{{\mathrm{haf}}}\nolimits}

\def\a{\alpha}
\def\b{\beta}
\def\g{\gamma}
\def\d{\delta}

\def\s{\sigma}

\def\la{\lambda}

\allowdisplaybreaks

\newcommand{\m}{\bs{m}}
\newcommand{\n}{\bs{n}}
\newcommand{\x}{\bs{x}}
\newcommand{\y}{\bs{y}}
\newcommand{\be}{\bs{e}}
\newcommand{\ba}{\bs{a}}

\DeclareSymbolFont{Eulerscripteusm10}{U}{eus}{m}{n}
\SetSymbolFont{Eulerscripteusm10}{bold}{U}{eus}{b}{n}
\DeclareMathSymbol{\euI}{\mathord}{Eulerscripteusm10}{"4A}
\DeclareMathSymbol{\rH}{\mathord}{Eulerscripteusm10}{"48}
\DeclareMathSymbol{\euK}{\mathord}{Eulerscripteusm10}{"4B}
\DeclareMathSymbol{\euR}{\mathord}{Eulerscripteusm10}{"52}
\DeclareMathSymbol{\rS}{\mathord}{Eulerscripteusm10}{"53}

\begin{document}

\title{Graph isomorphism and Gaussian boson sampling}

\author[K. Br\'adler]{Kamil Br\'adler}
\author[S. Friedland]{Shmuel~Friedland}
\author[J. Izaac]{Josh Izaac}
\author[N. Killoran]{Nathan Killoran}
\author[D. Su]{Daiqin Su}
\address[Kamil Br\'adler]{ORCA Computing, 84 Wood Lane, London,  W12 0BZ, UK; formerly Xanadu.
}
\address[Josh Izaac, Nathan Killoran,  Daiqin Su]{Xanadu, Toronto, Canada}
\address[Shmuel~Friedland]{Department of Mathematics, Statistics and Computer Science, University of Illinois, 851 South Morgan Street, Chicago, Illinois 60607-7045, USA}
\email{kamilbradler@gmail.com}
\email{friedlan@uic.edu}
\email{josh@xanadu.ai}
\email{nathan@xanadu.ai}
\email{sudaiqin@gmail.com}

\date{March 22, 2021}

\begin{abstract}
We introduce a connection between a near-term quantum computing device, specifically a Gaussian boson sampler, and the graph isomorphism problem. We propose a scheme where graphs are encoded into quantum states of light, whose properties are then probed with photon-number-resolving detectors.
We prove that the probabilities of different photon-detection events in this setup can be combined to give a complete set of graph invariants. Two graphs are isomorphic if and only if their detection probabilities are equivalent. We present additional ways that the measurement probabilities can be combined or coarse-grained to make experimental tests more amenable. We benchmark these methods with numerical simulations on the Titan supercomputer for several graph families: pairs of isospectral nonisomorphic graphs, isospectral regular graphs, and strongly regular graphs.
\end{abstract}

\maketitle

\noindent {\bf 2020 Mathematics Subject Classification}:  05C50, 05C60, 15A15, 68Q12, 81P68

\noindent \emph{Keywords}:  Gaussian boson sampling, graph isomorphism, hafnian, quantum GI algorithm, strongly regular graph

\thispagestyle{empty}
\section{Introduction}\label{sec:intro}

The problem of graph isomorphism (GI) lies at an interesting point in the landscape of computational complexity theory. Though algorithms have been recently proposed which run in `quasi-polynomial' time~\cite{babai2016graph, helfgott2017graph}, it is still an open question in theoretical computer science whether there exists a polynomial-time algorithm that can determine whether two graphs are isomorphic; indeed, graph isomorphism is likely to belong to the class of NP-intermediate computational problems. Two other well-known problems which have similar status in the complexity landscape are integer factoring and the discrete logarithm problem. Famously, while no classically efficient algorithm for these two problems is known, they can be solved in polynomial time on quantum computers~\cite{shor1994algorithms, shor1999polynomial}.
Quantum algorithms with a superpolynomial runtime advantage have also been proposed for linear systems~\cite{harrow2009quantum, childs2017quantum}, semidefinite programming~\cite{brandao2017quantum, brandao2017exponential}, knot invariants~\cite{freedman2002simulation, freedman2002modular, wocjan2006jones}, and partition functions~\cite{aharonov2007polynomial, van2009quantum, geraci2008exact}, among many others.  Boson sampling is a strong candidate to demonstrate the quantum computational advantage~\cite{aaronson2011computational}.  Zhong et al. measured a sampling rate that is about 1014-fold faster than using state-of-the-art classical simulation strategies and supercomputers \cite{quantadv}

Given these other success cases, it is natural to hypothesize that  may also be useful for the graph isomorphism problem.

Over the last several years, several works have explored this problem, with quantum algorithms for tackling graph isomorphism proposed based on quantum annealing~\cite{gaitan2014graph, zick2015experimental, calude2017qubo} and quantum graph states~\cite{mills2017proposal}. However, the bulk of quantum algorithm proposals to distinguish non-isomorphic graphs have utilized the time-evolution of a quantum walker to calculate `graph invariants' or `graph certificates' which, ideally, produce identical results for two graphs if and only if they are isomorphic. Of the algorithms proposed, they differ mainly in the number of particles involved, the presence of interactions, localised perturbations, and construction of the GI certificate~\cite{douglas2008classical,berry2011two,emms2009coined,wang2015graph,rudinger2012noninteracting}. It has subsequently been proven using this approach that conventional quantum walk algorithms, both discrete-time and continuous-time, with interactions and
perturbations, cannot distinguish arbitrary non-isomorphic graphs~\cite{rudinger2012noninteracting, rudinger2013comparing, mahasinghe2015phase}.

To test the distinguishing ability of proposed quantum GI algorithms, a common benchmark has become their capacity to distinguish nonisomorphic strongly regular graphs (SRGs) with the same graph parameters. This provides an analytic approach to investigate the effectiveness of graph isomorphism proposals; if a particular certificate will always fail to distinguish two non-isomorphic SRGs, this can be shown to be because all elements of a certificate, as well as their multiplicities, are functions of SRG family parameters~\cite{gamble2010two}.

In this work, we present an approach to graph isomorphism which uses a near-term quantum computational device, namely a photonics-based Gaussian boson sampling apparatus~\cite{hamilton2016gaussian, kruse2018detailed}. For this method, graphs are encoded into quantum-optical states of light -- specifically Gaussian states -- which are then subjected to photon-number-resolving measurements. Mathematically, we show that the resulting measurement outcome probabilities can be combined to give a complete set of graph invariants. Two graphs are isomorphic if and only if these graph invariants are equal. We also present several ways that these measurement probabilities can be combined and coarse-grained to obtain new quantities which can be used to distinguish nonisomorphic graphs. Finally, we perform classical numerical simulations of our proposed method on the Titan supercomputer. Using these results, we are able to distinguish 3852 out of 3854 nonisomorphic graphs using only a subset of measurement events. The remaining two graphs were distinguished by failing to satisfy a necessary condition introduced here as well.

\section{Main results summarized}\label{sec:summary}

Our main result is a necessary and sufficient condition to distinguish isospectral nonisomorphic graphs by virtue of comparing the probabilities of the measurement patterns of the graphs encoded in a Gaussian boson sampling (GBS) apparatus.  We discovered the vital role played by a matrix function  called the hafnian~\cite{caianiello1953quantum}, applied to an adjacency matrix,  for the GI problem. It leads to a complete set of graph invariants. The hafnian belongs to the family of matrix functions such as the determinant, permanent and pfaffian~\cite{barvinok2017combinatorics}.  It has been established that photon detection probabilities can be expressed in terms of the hafnians of a collection of graphs related to the original graph~\cite{bradler2017gaussian}. Multiphoton detection probabilities are handled by introducing a new matrix product related to the Kronecker product and by showing how the output probabilities depend on the hafnian of the graph adjacency matrix as well. We further strengthen our graph invariant results by introducing the so-called symmetrized graphs invariants and showing that they correspond to coarse-grained measurement events in GBS. The measurement events are given by the stratification according to the total photon number and partitioned into the orbits of the permutation group. Their hafnian-based coarse-grained probabilities are again sufficient to distinguish isospectral nonisomorphic graphs. We extend these insights by deriving necessary conditions for isospectral graphs to be isomorphic by comparing the coarse-grained partition-averaged photon distribution from the Gaussian boson sampler.

Our method differs from  previous quantum GI algorithm proposals. A great majority have utilized quantum walks, either using discrete-time quantum walks (DTQWs)~\cite{emms2009coined,berry2011two} or continuous-time quantum walks (CTQWs)~\cite{mahasinghe2015phase,gamble2010two,rudinger2012noninteracting,rudinger2013comparing}. Although the graph invariants constructed via quantum walk propagation on graph structures have shown success in distinguishing various families of SRGs, it has been proven that this distinguishing power is not universal --- there will always exist graphs which (current) quantum walk-based algorithms cannot distinguish~\cite{smith2012algebraic}. In order to execute a quantum walk-based algorithm in a universal quantum photonic platform, it is necessary to implement a non-Gaussian  operation as a vertex-dependent shift or via multiple interacting walkers. This is a major obstacle with the current and near-term technology our proposal does not suffer from. GBS is a Gaussian circuit followed by an array of photon-number-resolving detector (PNR) representing a non-Gaussian element. Unlike non-Gaussian unitary transformations, the PNRs are available in the state-of-the-art laboratories.

The most comprehensive simulations of quantum methods for GI were performed in~\cite{douglas2008classical} and~\cite{rudinger2012noninteracting}. We successfully tested three types of isospectral graphs: pairs of isospectral nonisomorphic graphs (PINGs) as the first examples of such graphs~\cite{baker1966drum}, isospectral regular graphs~\cite{little2006combinatorial} and mainly SRGs. There are numerous resources available detailing the SRG families containing more than one non-isomorphic graph~\cite{spence2018,brouwer2017}; as a result, SRGs have become a common benchmark in studying the distinguishing powers of the GI algorithms. Note that there may be other graph classes (such as $k$-equivalent graphs) which have been proven to be harder to distinguish than strongly regular graphs for particular quantum GI algorithms~\cite{smith2012algebraic} --- however, SRGs remain an ideal testing set, simply due to the large number of relatively small non-isomorphic graphs present in specific families~\cite{spence2018,brouwer2017}. The largest tested and distinguished family by our approach was SRG(35,18,9,9) containing 3854 isospectral graphs. This family is supposedly tested in~\cite{rudinger2012noninteracting}. However, the size of the family is mistakenly taken to be only 227~graphs (see Table~I.). The same error appears in~\cite{douglas2008classical}. Ironically, another SRG family considered there (SRG(35,16,16,8)), that happens to be complementary to SRG(35,18,9,9) and thus containing 3854 graphs as well, is counted properly and analyzed (see Table~1.).

Section~\ref{sec:notation} contains all necessary definitions and previous results used in the paper including a detailed GBS description and a formal introduction of SRGs. Section~\ref{sec:main} contains the main result and is split into four subsections: In~\ref{subsec:suppresults} we gather several supporting  results followed by the main results in Sections~\ref{subsec:graphinv},~\ref{subsec:hierarchyTower} and~\ref{subsec:meanphoton}. Section~\ref{sec:simulations} contains the simulation results and Section~\ref{sec:discussion} concludes with a scalability discussion and other open questions.  In Appendix~\ref{sec:experiment} we informally introduce the hardware setup (Gaussian boson sampler) where studied graphs are encoded.  In Appendix~\ref{sec:algo} we present the GBS quantum GI algorithm applied to various SRG families and other isospectral graphs.  In Appendix~\ref{sec:symbols} we summarize with a table the most important symbols and their meaning.

\section{Notation and preliminaries}\label{sec:notation}

In the following text the symbol $\bbJ_{k,\ell}$ denotes an all-ones rectangular matrix of size $k\times\ell$ and $\bbJ_{k}\equiv\bbJ_{k,k}$. The following notation is extensively used: $\partial^n_x\equiv\partial_{x,\dots,x}={\partial^n\over\partial x^n}$ and $\partial^{n_i}_{x_i,\ol x_i}={\partial^{n_i}\over\partial x_i^{n_i}}{\partial^{n_i}\over\partial\ol{x}_i^{n_i}}$. Letting $\n=(n_1,\dots,n_M),\x=(x_1,\dots,x_M)$, we occasionally write $\prod_{i=1}^{M}\partial^{n_i}_{x_i}=\partial^{|\n|}_{\x}$ and $\prod_{i=1}^{M}\partial^{n_i}_{x_i,\ol x_i}=\partial^{|\n|}_{\x,\ol\x}$. The symbol~$\df$ stands for `defined' and a positive-definite matrix $A$ will be denoted by $A\succ 0$.  Recall that any Gaussian $n$-dimensional real distribution with  zero mean, denoted as $G_{\Sigma}$, is given by
\[
\frac{1}{(2\pi)^{\frac{n}{2}}\sqrt{\det \Sigma}} \exp{[-\frac{1}{2}\bs{x}^\top\Sigma^{-1}\bs{x}]}.
\]
Here, $\Sigma$ is a positive definite matrix which is the covariance matrix of the Gaussian variables $X_1,\ldots,X_n$.

Since $\Sigma$ is positive definite, there exists a unique positive definite matrix $A$ such that $\Sigma=A^2$.  Let us change the variables $\bs{y}=A^{-1}\bs{x}$. That is, $\bs{x}=A\bs{y}$.  Hence the determinant of the Jacobian is $\det{A}$. As $\det{\Sigma}=(\det{A})^2 $ we get that the density distribution for $(Y_1,\ldots,Y_n)$ is the standard normal density distribution $\frac{1}{(2\pi)^{\frac{n}{2}}}\exp{[-\frac{1}{2}\bs{y}^\top\bs{y}]}$.  Therefore $Y_1,\ldots,Y_n$ are independent standard random variables. Assume that $A=[a_{ij}]$ is a positive definite symmetric matrix.  Then
\[
X_i=\sum_{j=1}^n a_{ij}Y_j, \quad i\in[n]
\]
and
\[
\bbE{[X_iX_j]}=\bbE\big[\big(\sum_{p=1}^n a_{ip}Y_p\big)\big(\sum_{q=1}^n a_{jq}Y_q\big)\big]
=\sum_{p,q=1}^n a_{ip}a_{iq} \bbE[Y_pY_q]=\sum_{p=1}^n a_{ip}a_{jp}=\Sigma_{ij} .
\]
Observe the well known fact that the odd moments $\bbE[\prod_{i=1}^n X^{m_i}]$, where $(m_1,\ldots,m_n)\in\bbZ_+^n$ and $\sum_{i=1}^n m_i$ is odd, are zero.  A polynomial $p(\bs{x}), \bs{x}=(x_1,\ldots,x_n)\in\bbR^n$ is called symmetric if for each permutation $\sigma:[n]\to[n]=(1,\ldots,n)$ we have the equality $p(\bs{x})=p(\bs{x}_{\sigma})$, where $\bs{x}_{\sigma}=(x_{\sigma(1)},\ldots,x_{\sigma(n)})$.

Denote by $\mathfrak{S}_n$ the symmetric group of bijections $\sigma:[n]\to[n]$.  Denote by $\Pcal_n\subset \bbR^{n\times n}$ the group of $n\times n$ permutation matrices.  So $P(\sigma)\bs{x}=\bs{x}_{\sigma}$.

Recall that two square matrices $A,B$ are permutationally similar, if $B=PAP^\top$, where $P$ is a permutation matrix. In this case $P^{-1}=P^\top$. Two Gaussian distributions corresponding to positive definite covariance matrices $\Sigma,\Sigma'\in \bbR^{n\times n}$ are called isomorphic, if $\Sigma'=P^\top\Sigma P$ for some permutation $P=P(\sigma)$.  That is $\bs{x}^\top (\Sigma')^{-1}\bs{x}=\bs{x}_{\sigma}^\top (\Sigma)^{-1} \bs{x}_{\sigma}$, where $\sigma\in\mathfrak{S}_n$, for all $\bs{x}\in\bbR^n$.

Denote by $\rH_{N}\supset \rH_{+,N}\supset \rH_{++,N}$ the real space of $N\times N$ hermitian matrices, the closed cone of positive semidefinite hermitian matrices, and the open set of positive definite hermitian matrices. For $F\in \rH_N$ denote by $\lambda_1(F)\ge \cdots\ge \lambda_N(F)$ the $N$ eigenvalues of $F$. Recall that the spectral norm  of $F$ is given by $\|F\|_2=\max(\lambda_1(F),-\lambda_N(F))$. For $X,Y\in\rH_N$ we denote $X\preceq Y$ and  $X\prec Y$ if $Y-X\in\rH_{+,N}$ or $Y-X\in\rH_{++,N}$, respectively.

\subsection{Gaussian Boson Sampling}

\begin{defi}\label{def:haf}
  Let $C = [c_{ij}]\in\bbR^{2M\times2M}$ be a symmetric matrix and let $\Mcal_{2M}$ denote all partitions $\varsigma$ to unordered disjoint pairs. Then
  \begin{equation}\label{eq:hafA}
        \haf{C}\df\sum_{\varsigma\in\Mcal_{2M}}\prod_{(uv)\in\varsigma}c_{uv}
  \end{equation}
  is the hafnian of~$C$~\cite{caianiello1953quantum}.
\end{defi}
Given a detection event $\n$, its measurement probability was shown in~\cite{hamilton2016gaussian} to be
\begin{equation}\label{eq:ProbMixedGBS}
  p(\n)={1\over{\n!}\sqrt{\det{\s_Q}}}
  \partial^{|\n|}_{\bs{\b},\ol{\bs\b}}
  e^{{1\over2}\bs{\g}^\top C\bs{\g}}\big\rvert_{\bs{\g}=0},
\end{equation}
where $\n!=n_1!\times\dots\times n_M!$ and $\bs{\g}\df(\bs{\b},\ol{\bs{\b}})=(\b_1,\dots,\b_M,\ol\b_1,\dots,\ol\b_M)\in\bbC^{2M}$ which we view as a column vector (even though $\ol{\bs{\b}}$ is a complex conjugate of a complex $\bs\b$ (entrywise), we consider $\ol\b_i$ as a new variable). We denote
\begin{equation}
X_{2M}=\begin{bmatrix}
           0 & \bbI_M \\
           \bbI_M & 0 \\
         \end{bmatrix}.
\end{equation}
Then
\begin{equation}\label{eq:sigmaQ}
\s_Q=(\bbI_{2M}-X_{2M}C)^{-1}
\end{equation}
and $\s = \s_Q - \bbI_{2M}/2$ is the covariance matrix in Eq.~\eqref{eq:covmatrix}.

Note that in order for $\s$ to be a covariance matrix, $C$ has to satisfy certain restrictions which will be elaborated on later.
We call a GBS detection event corresponding to the measurement pattern $\bs{n}\df(n_1,\dots,n_{M})$ of an $M$-mode matrix $C$ of size $2M$ \emph{pure} if $n_i=n_j,\forall i,j$ and \emph{mixed} otherwise. For $\n=(1,\dots,1)$ Eq.~\eqref{eq:ProbMixedGBS} reduces to~\eqref{eq:p11111}~\cite{hamilton2016gaussian}.

\subsection{Graphs: eigenvalues, isospectral graphs, strongly regular graphs and graph isomorphism}

We now recall briefly some well known results on graphs that we use in this paper.
An undirected simple graph consist of a set of $n$-vertices $V=\{v_1,\ldots,v_n\}$ which we identify with $[n]$.  The set of edges $E(G)=E=\{e_1,\ldots,e_m\}$ is the set of unordered pairs $e=(u,v)$ where $u\ne v\in V$.
We say that $e$ connects $u$ and $v$, or $e$ is adjacent to $u$ and $v$.   A simple path in $G$ is an ordered subset of $E$: $(u_1,u_2),\ldots,(u_k,u_{k+1})$.  A graph $G$ is called connected if for each pair of vertices $u\ne v$ there is simple path such that $u_1=u$ and $u_{k+1}=v$.

A compete graph
on $n$-vertices is denoted by $K_n$, so the cardinality of its edges is $|E(K_n)|=n(n-1)/2$.  Given a simple graph $G$ on $n$-vertices then $E(G)$ is a subset of $E(K_n)$.   Hence $m\le n(n-1)/2$.
The complement of $G$, denoted as $G^c$, is a graph on $[n]$ vertices with the edges $E(G^c)=E(K_n)\setminus E(G)$.  Thus the complement of $K_n$ is a graph with no edges, the null graph.  Given a subset $W\subset V$ the induced subgraph $G(W)=(W, E(G(W))$, where $E(G(W))$ is the subset of edges in $V$ that connect two vertices in $W$.   A clique is a subgraph  $G(W)$ which is a complete graph on $W$.  A graph $G$ is called regular (or $k$-regular) if each vertex $v\in V$ is adjacent to exactly $k$-edges.  A graph
$G$ is called bipartite if $V$ is a union of disjoints subsets of vertices $V_1,V_2$ where the edges $E$ connect vertices in $V_1$ to vertices in $V_2$.

The adjacency matrix $A(G)=A=[a_{ij}]$ is an $n\times n$ symmetric matrix with zero diagonal whose off diagonal entries are zero or one.   Then $a_{ij}=1$ if and only if $(i,j)\in E(G)$.  If we relabel the vertices of $G$, i.e., apply a bijection $\sigma: [n]\to [n]$, then the  new adjacency matrix $\tilde A$ is $PAP^\top$ for some permutation matrix $P\in \Pcal_n$.  As $A$ is real symmetric it has $n$ real eigenvalues, counted with their multiplicities: $\lambda_1(G)=\lambda_1\ge \cdots\ge \lambda_n(G)=\lambda_n$.  As $A$ is a nonnegative matrix the Perron-Frobenius theorem yields that $\lambda_1\ge |\lambda_n|$.  Equality holds is and only if $G$ is bipartite.
If $G$ is connected then $\lambda_1>\lambda_2$, that is, $\lambda_1$ is a simple eigenvalue of $G$.  Furthermore, if $G$ is a $k$-regular graph then $\lambda_1=k$.
Since $\Pcal_n$ is a subgroup of the group of orthogonal matrices, it follows that the eigenvalues of $G$ do not depend on the labeling of the vertices of $G$.

A graph $H$ on $n$-vertices is called isomorphic to $G$, if $H$ is a relabeling  $G$.  That is, if $A(H)=PA(G)P^\top$ for some $P\in\Pcal_n$.  Thus a necessary conditions
for two graph to be isomorphic is to be isospectral, i.e., to have the same sequence of eigenvalues.   That is, $A(G)$ and $A(H)$ have the same characteristic polynomial.
Hence we can check in polynomial time if $G$ and $H$ are isospectral.
As we pointed out in Introduction, the problem of graph isomorphism (GI) lies at an interesting point in the landscape of computational complexity theory.

In studying the graph isomorphism problem, it is convenient to consider a class of graphs known to be classically intractable to distinguish. An important tractable feature is the graph eigenvalues and the first examples of isospectral graphs were  pairs of isospectral nonisomorphic graphs (PINGs)~\cite{baker1966drum}. The smallest connected example of a PING is on six vertices, see Fig.~\ref{fig:PING6v}.
    \begin{figure}[h]
      \resizebox{8cm}{!}{\includegraphics{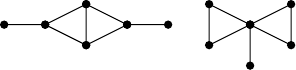}}
      \caption{ PING on six vertices. }
      \label{fig:PING6v}
    \end{figure}
PINGs may have some tractable features enabling one to easily decide that they are not isomorphic.

One of best known positive results in graph isomorphism is the following result of  Babai-Grigoryev-Mount \cite{BGM}:  Let $k$ be a fixed integer.  Assume that $G$ and $H$ are isospectral, and the multiplicity of each eigenvalue is at most $k$.  Then there exists a polynomial time algorithm that decides if $G$ and $H$ are isomorphic.

The situation gets complicated for graphs with symmetries such as strongly regular graphs (SRGs), defined as follows~\cite{godsil2001strongly}.
\begin{defi}
	Let $G(V,E)$ be a regular graph of degree $k$ consisting of $N$ vertices and adjacency matrix $A$, that is neither a complete graph ($A\neq\bbJ_N-\bbI_N$) nor a null graph. $G$ is then said to be strongly regular with parameters SRG$(N,k,\lambda,\mu)$ if every pair of adjacent vertices have exactly $\lambda$ common neighbours, and every pair of \textit{non}-adjacent vertices have exactly $\mu$ common neighbours.
\end{defi}
Recall also that an SRG graph has three distinct eigenvalues: simple eigenvalue $\lambda_1=k$, and other two eigenvalues with high multiplicity: at least one eigenvalue of multiplicity  at least $(n-1)/2$.

\begin{lem}
	Let $v_i\in V$ be a vertex in SRG$(N,k,\lambda,\mu)$. Then $k(k-\lambda-1)=\mu(N-k-1)$. Thus the SRG parameters are not independent.
\end{lem}
\begin{proof}
  Consider a vertex $v$ in a graph with parameters $SRG(N,k,\lambda,\mu)$, and count in two ways the number of edges from vertices adjacent to $v$ to vertices non-adjacent to $v$.
\end{proof}
If multiple non-isomorphic strongly regular graphs share the same set of SRG parameters, we refer to this graph set as an \emph{SRG family}, often simply denoted by the SRG parameters $(N,k,\lambda,\mu)$. Graphs within the same SRG family share various properties that are dependent only on the SRG parameters. The spectrum  is rather special, consisting of just three eigenvalues with known multiplicities.  SRG families contain non-isomorphic graphs which are isospectral, and difficult to distinguish using common classical measures~\cite{godsil2001strongly}.

\section{Graph isomorphism via Gaussian boson sampling}\label{sec:main}

\subsection{Multiphoton contributions in GBS -- supporting results}\label{subsec:suppresults}
We start with the exploration of how to interpret Eq.~\eqref{eq:ProbMixedGBS} for the detection events $\n$ where $n_i>1$ for some $i$. This corresponds to a multiphoton contribution of the output probability function. The multiphoton contributions play a vital role in our analysis.
\begin{defi}\label{def:sox}
Assume that $A=[a_{ij}]\in \bbR^{m_1\times m_2}$ and in total $m_1\times m_2$ matrices $B=[B_{ij}]$ where each $B_{ij}$ is an $n_i\times n_j$ matrix.  Then the \emph{reduced Kronecker product} $C=A\sox B$ will denote a block matrix $C$ partitioned as $B$ and the blocks are $C=[a_{ij}B_{ij}]$.
\end{defi}
\begin{rem}
  The dimension of $A\sox B$ is the dimension of the matrix $B$.  Then $A\sox B$ is a submatrix of the Kronecker tensor product of two matrices $A\otimes B=[a_{ij}B]$.   The matrix $B$ in this paper will always be assembled from $B_{ij}=\bbJ_{n_i,n_j}\in\bbR^{n_i\times n_j}$ where $\n$ is a measurement pattern. In this case we will write $B=\bbJ_{|\n|}\in\bbR^{|\n|\times|\n|}$ where $|\n|=\sum_{i=1}^{M}n_i$. Note that if $\n$ is a pure event then the reduced Kronecker product $\sox$ becomes the ordinary Kronecker product $A\ox\bbJ_{n_1}$. Also note that if $\dim{B_{ij}}=1,\forall i,j$ the reduced Kronecker product becomes the Hadamard (Schur) product.
\end{rem}
\begin{exa}
  Let $\n=(3,2,1,4)$ and $A$ an adjacency matrix of a simple weighted graph (without loops). Then
\[
A\sox\bbJ_{|\n|}=\left[\begin{array}{@{}c|c|c|c@{}}
\raisebox{-1.2pt}{\mbox{\Large 0}}
&
  \begin{matrix}
  a_{12} & a_{12} \\
  a_{12} & a_{12} \\
  a_{12} & a_{12}
    \end{matrix}
&
  \begin{matrix}
  a_{13}  \\
  a_{13}  \\
  a_{13}
    \end{matrix}
&
  \begin{matrix}
  a_{14} & a_{14} &  a_{14} & a_{14} \\
  a_{14} & a_{14} &  a_{14} & a_{14} \\
    a_{14} & a_{14} &  a_{14} & a_{14}
    \end{matrix}
\\\hline
      \begin{matrix}
        a_{12} & a_{12} & a_{12} \\
        a_{12} & a_{12} & a_{12} \\
      \end{matrix}
      &
      \raisebox{-1.2pt}{\mbox{\Large 0}}
    &
      \begin{matrix}
      a_{23}  \\
      a_{23}
        \end{matrix}
    &
      \begin{matrix}
      a_{24} & a_{24} &  a_{24} & a_{24} \\
      a_{24} & a_{24} &  a_{24} & a_{24}
        \end{matrix}
        \\\hline
       \begin{matrix}
        a_{13} & a_{13} & a_{13}
      \end{matrix}
    &
      \begin{matrix}
        a_{23} & a_{23} \\
      \end{matrix}
    &
      \raisebox{-1.2pt}{\mbox{\Large 0}}
    &
      \begin{matrix}
      a_{34} & a_{34} &  a_{34} & a_{34}
        \end{matrix}
\\\hline
      \begin{matrix}
  a_{14} & a_{14} &  a_{14}  \\
  a_{14} & a_{14} &  a_{14} \\
    a_{14} & a_{14} &  a_{14} \\
    a_{14} & a_{14} &  a_{14}
    \end{matrix}
        &
      \begin{matrix}
      a_{24} & a_{24}  \\
      a_{24} & a_{24} \\
      a_{24} & a_{24}  \\
      a_{24} & a_{24}
        \end{matrix}
    &
  \begin{matrix}
  a_{34}  \\
  a_{34}  \\
  a_{34}  \\
  a_{34}
    \end{matrix}
      &
      \raisebox{-1.2pt}{\mbox{\Large 0}}
\end{array}\right].
\]
\end{exa}
The reason for introducing a new kind of structure is a compact expression for the probability of measurement of a mixed multiphoton event $\n$ as a hafnian function not unlike~Eq.~\eqref{eq:p11111} for $\n=(1,\dots,1)$.
\begin{defi}\label{def:GBSenc}
A $2M\times 2M$-dimensional real symmetric matrix $R$ will be called \emph{GBS encodable} if we can find a covariance matrix $\s_Q$ such that
\begin{equation}
 R = X_{2M}\big( \bbI_{2M} - \sigma_Q^{-1} \big).
\end{equation}
\end{defi}
Ref.~\cite{bradler2017gaussian} introduced a necessary criterion for $R$ to be GBS encodable. For some real symmetric $\tilde R$ not satisfying the conditions a general procedure was created to produce a matrix related to $\tilde R$ that is GBS-encodable. It consists of taking $\tilde R\mapsto R\df c(\tilde R\oplus \tilde R)$ where $0<c<1/\|\tilde R\|_2$.

Even though this procedure is always guaranteed to succeed in creating a Gaussian covariance matrix, it is not a necessary condition. Here we strengthen this previous result by loosening the requirements on~$R$.
\begin{lem}\label{lem:posdefsigA}
Let $G\in \rH_{N}$ and assume that $G=(\bbI_N-F)^{-1}-\frac{1}{2}\bbI_N$.  Then
\begin{enumerate}
\item $G\succ 0$ if and only if $\|F\|_2<1$.
\item  Suppose that $N=2M$, $F=\begin{bmatrix}F_{11}&F_{12}\\F_{21}&F_{22}\end{bmatrix}$.  Then the following conditions hold
\begin{equation}\label{eq:posdefsigA1}
G\succ 0, \quad  G+ \frac{1}{2}\begin{bmatrix}\bbI_M&0\\0&-\bbI_M\end{bmatrix}\succeq 0
\end{equation}
if and only if $\|F\|_2<1$ and $F_{22}\succeq 0$.
\end{enumerate}
\end{lem}
\begin{proof}\mbox{}
\begin{enumerate}
    \item Clearly we need to have that $(\bbI_N-F)^{-1}\succ 0$. This equivalent to $\lambda_1(F)<1$.  The assumption that $(\bbI_N-F)^{-1}\succ \frac{1}{2}\bbI_N$ is equivalent to $\lambda_N((\bbI_N-F)^{-1})>\frac{1}{2}$.  This is equivalent to $\lambda_N(F)>-1$.  Therefore the claim follows.
    \item Use~(1) to get the assumption that  $G\succ 0$ is equivalent to  $\|F\|_2<1$.  The assumption that $G+ \frac{1}{2}\begin{bmatrix} \bbI_M&0\\0&-\bbI_M\end{bmatrix}\succeq 0$ is equivalent to $(\bbI_N-F)^{-1}\succeq\begin{bmatrix} 0&0\\0&\bbI_M\end{bmatrix}$.  This is equivalent to
        $$
        \bbI_N\succeq (\bbI_N-F)^{\frac{1}{2}}\begin{bmatrix}0&0\\0&\bbI_M\end{bmatrix}(\bbI_N-F)^{\frac{1}{2}}\succeq 0.
        $$
        This inequality is equivalent to
        \[
        1\ge \la_1\Big((\bbI_N-F)^{\frac{1}{2}}\begin{bmatrix}0&0\\0&\bbI_M\end{bmatrix} (\bbI_N-F)^{\frac{1}{2}}\Big)=\lambda_1\Big((\bbI_N-F)\begin{bmatrix} 0&0\\0&\bbI_M\end{bmatrix}\Big).
        \]
        Observe next that
        \[
        (\bbI_N-F)\begin{bmatrix} 0&0\\0&\bbI_M\end{bmatrix}=\begin{bmatrix} 0&-F_{12}\\0&\bbI_M-F_{22}\end{bmatrix}.
        \]
        Hence
        \[
        1\ge \lambda_1\Big((\bbI_N-F)\begin{bmatrix} 0&0\\0&\bbI_M\end{bmatrix}\Big)=1-\lambda_M(F_{22})
        \]
        is equivalent to $\lambda_M(F_{22})\ge 0$, that is $F_{22}\succeq 0$.
\end{enumerate}
\end{proof}
\begin{cor}\label{cor:fromadjm}  Let  $R\in\bbR^{2M\times2M}$ be a nonzero real symmetric matrix with the following partition to $M\times M$ blocks: $R=\begin{bmatrix}
    R_{11}&R_{12}\\
    R_{21}&R_{22}
\end{bmatrix}$.
Then there exists a Gaussian covariance matrix $\sigma$ such that
$cR=X_{2M}[\bbI_{2M}-(\sigma+\frac{1}{2}\bbI_{2M})]$ if and only if:
\begin{enumerate}
\item $R_{11}=R_{22}$ and $R_{12}=R_{21}$.
\item $R_{12}\succeq 0$.
\item $c\in (0,\frac{1}{\|R\|_2})$.
\end{enumerate}
\end{cor}

Part \emph{(3)} follows from the observation the  singular values of $R$ and $X_{2M}R$ are the same.  Hence $\|R\|_2=\|X_{2M} R\|_2$.
\begin{rem}
  The previous lemma was presented for the sake of completeness. In the rest of the paper we use the construction of GBS-encodable matrices introduced already in~\cite{bradler2017gaussian}.
\end{rem}

For future reference we recall the following straightforward result~\cite{bradler2017gaussian}.
\begin{lem}\label{lem:hafsquared}
  Let $M$ be even and $C=A\op A$ a real symmetric matrix of dimension $2M\times2M$. Then
  \begin{equation}\label{eq:hafdoubled}
    \haf{[c(C+k\bbI_{2M})]}=c^{M}\haf^{\,2}{A},
  \end{equation}
  where $c>0$ and $k\in\bbR$.
\end{lem}
\begin{lem}
  Assume   $C = [c_{ij}]\in\bbR^{2M\times2M}$ is a symmetric and GBS-encodable matrix. Then the probability of sampling in the GBS event, Eq.~\eqref{eq:ProbMixedGBS}, can be expressed as
  \begin{equation}\label{eq:ProbMixed}
  p(\bs{n})={1\over{\n!}\sqrt{\det{\s_Q}}}\haf{[C\sox\bbJ_{2|\n|}]}.
  \end{equation}
\end{lem}
\begin{proof}
  To prove this equality we assume that $C=A\op A$ where $A$ is an arbitrary symmetric matrix of order $M$. Let $\tilde{p}(n)$ be defined by the right-hand side of~\eqref{eq:ProbMixedGBS}. We will show that $\tilde{p}(\n)=p(\n)$. Consider $\bbJ_{2|\n|}$ as a $2M\times2M$ block matrix $[\bbJ_{n_i,n_j}]$ where $n_{i+M}=n_i$. Hence $C\sox\bbJ_{2|\n|}=A\sox\bbJ_{|\n|}\op A\sox\bbJ_{|\n|}$. We now consider the quadratic form $\bs{\g}^\top(C\sox\bbJ_{2|\n|})\bs{\g}$ and focus on $\bs{\b}^\top(A\sox\bbJ_{|\n|})\bs{\b}$ (the `second half' is treated in exactly the same way). We substitute
  \begin{equation}\label{eq:betaToalphan}
    \b_i\mapsto(\a_{i1},\dots,\a_{in_i}).
  \end{equation}
  in~\eqref{eq:ProbMixedGBS}.  We define its action for the quadratic form to be
\begin{equation}\label{eq:bAb}
  \bs{\b}^\top A\bs{\b}=\sum_{ij}\b_i\b_ja_{ij}\mapsto
  \sum_{ij}(\a_{i1}+\dots+\a_{in_i})(\a_{j1}+\dots+\a_{jn_j})a_{ij}
  =\bs{\a}^\top[a_{ij}\bbJ_{n_in_j}]\bs{\a}\equiv\bs\a^\top(A\sox\bbJ_{|\n|})\bs\a,
\end{equation}
  where we used Def.~\ref{def:sox} by setting $B_{ij}=\bbJ_{n_i,n_j}$. We further set $\partial_{\b_i}^{n_i}=\partial_{\a_{i1},\dots,\a_{in_i}},\forall i$ (and similarly for the conjugated variables $\ol\b_i$ and the corresponding $\ol\a_{in_i}$) and write
  \begin{equation}\label{eq:ProbMixedtilde}
  \tilde{p}(\bs{n})={1\over{\bs{n}!}\sqrt{\det{\s_Q}}}\prod_{k=1}^{M}\partial_{\a_{k1},\dots,\a_{kn_k}}
  e^{{1\over2}\bs\a^\top(A\sox\bbJ_{|\n|})\bs\a}\Big\rvert_{\bs{\a}=0}.
  \end{equation}
  It remains to show that $p(\bs{n})=\tilde{p}(\bs{n})$. Indeed, the higher-order partial derivatives in~\eqref{eq:ProbMixedGBS} result in the same expression as the first-order ones in~\eqref{eq:ProbMixedtilde} whenever we set $\bs{\a}=\bs{\b}=0$ at the end of the calculation. This is a consequence of the elementary properties of the differential operator, namely,
  \begin{equation}\label{eq:PDequality}
    \partial_{x_{k1},\dots,x_{k1}}f({\textstyle \sum}^{n_k}_{\ell=1}x_{k\ell})
    =\partial_{x_{k1},\dots,x_{kn_k}}f({\textstyle\sum}^{n_k}_{\ell=1}x_{k\ell} ),\ \forall k
  \end{equation}
  and the chain rule for the $n$-th derivative given by F{}a\`a di Bruno's formula for $(f\circ g)^{(n)}(x)$ in the special case of $g(x)=x+K$ where $K$ is a constant:
  \begin{equation}\label{eq:FaaDB}
    f^{(n)}(x+K)=h(x+K)
  \end{equation}
  whenever $f^{(n)}(x)\df h(x)$. Now we put all the pieces together. The RHS of~\eqref{eq:PDequality} is identified with~\eqref{eq:ProbMixedtilde} through $(x_{k1},\dots,x_{kn_k})\mapsto(\a_{k1},\dots,\a_{kn_k})$ (or its conjugate) for a given $1\leq k\leq M$ and so $f({\sum_{k=1}^M\sum}^{n_k}_{\ell=1}\a_{k\ell})=e^{{1\over2}\bs\a^\top(A\sox\bbJ_{|\n|})\bs\a}$. But then, according to the LHS of~\eqref{eq:PDequality} we may write~\eqref{eq:ProbMixedtilde} as
  \begin{equation}\label{eq:ProbMixedtilde1}
  \tilde{p}(\bs{n})={1\over{\bs{n}!}\sqrt{\det{\s_Q}}}\prod_{k=1}^{M}\partial_{\a_{k1},\dots,\a_{k1}}
  e^{{1\over2}\bs\a^\top(A\sox\bbJ_{|\n|})\bs\a}\Big\rvert_{\bs{\a}=0}.
  \end{equation}
  The RHS of~\eqref{eq:ProbMixedtilde1} is identified with the LHS of~\eqref{eq:FaaDB} by setting $x=\a_{k1},n=n_k$ and $K=\sum^{n_k}_{\ell=2}\a_{k\ell}$ for a given~$k$. Since
  $$
  h(\b_1)={\dif^n\over\dif\b_1^n}e^{{1\over2}\bs{\b}^\top A\bs{\b}}
  $$
  forms $p(\bs{n})$ and from~\eqref{eq:ProbMixedtilde} we get
  $$
  \prod_{k=1}^{M}\partial_{\a_{k1},\dots,\a_{kn_k}}
  e^{{1\over2}\bs\a^\top(A\sox\bbJ_{|\n|})\bs\a}\Big\rvert_{\bs{\a}=0}=\haf{[A\sox\bbJ_{|\n|}]}
  $$
  we may conclude that $\tilde{p}(\bs{n})=p(\bs{n})$ due to $f^{(n)}(\b_{1})|_{\b_{1}=0}=f^{(n)}(\a_{k1}+K)|_{\a_{k1}=K=0}$. This follows  from~\eqref{eq:FaaDB} and the definition of $h(x)$.
\end{proof}
Interestingly, many of the detection  events have probability zero:
 \begin{lem}\label{lem:nTOT}
  Let $C=c(A\op A)\in\bbR^{2M\times2M}$ and let $\n$ be a detection event where $|\n|=\sum_{i=1}^{M}n_i$. If there is any $n_i>|\n|/2$ then $p(\n)=0$.
\end{lem}
\begin{proof}
  Here we may assume $c=1$ and so $C=A\op A$. The probability expression (Eq.~\eqref{eq:ProbMixed}) contains $\haf{[A\sox\bbJ_{|\n|}]}$ and if $n_i>|\n|/2$ then $A\sox\bbJ_{|\n|}$ contains a  zero matrix of size greater than $\lfloor M/2\rfloor$ (placed in the lower right corner of $A\sox\bbJ_{|\n|}$), see Definition~\ref{def:sox}. The hafnian of such a matrix is zero. This follows from the hafnian definition (Def.~\ref{def:haf}) where the hafnian of a $2M\times2M$ matrix is a sum of products of $M$ entries $a_{ij}$. Since $i<j$ and none of the indices repeats for any summand then inevitably at least one of the $a_{ij}$'s in every summand equals zero.
\end{proof}

We prove several useful properties of the reduced Kronecker product $\sox$.
\begin{lem}\label{lem:cTimesSox}
  Let $C =[c_{ij}]\in\bbR^{2M\times2M}$ and $c\in\bbR$. Then
  $$
  \haf{[(cC)\sox\bbJ_{2|\n|}]}=c^{|\n|}\haf{[C\sox\bbJ_{2|\n|}]}
  $$
  where $|\n|=\sum_{i=1}^Mn_i$.
\end{lem}
\begin{proof}
    From Def.~\ref{def:sox} we find $(cC)\sox\bbJ_{2|\n|}=c(C\sox\bbJ_{2|\n|})$. This trivially follows from $[(cc_{ij})B_{ij}]=[c(c_{ij}B_{ij})]$. We then observe that $\dim{[C\sox\bbJ_{2|\n|}]}=2|\n|$ and the claim follows.
\end{proof}
\begin{rem}
  For $n_i=1,\forall i$ we recover $\haf{[cC]}=c^M\haf{C}$ since $|\n|=M$.
\end{rem}
\begin{lem}\label{lem:linearity}
  Let $A,B,C$ be matrices and assume that $A\sox C$ and $B\sox C$ are defined. Then $(A\oplus B)\sox (C\oplus C)=(A\sox C)\oplus (B\sox C)$.
  \begin{proof}
    The ordinary Kronecker product satisfies $(A\op B)\ox C'=A\ox C'\op B\ox C'$. By removing  columns and the corresponding rows from $A,B$ on both sides of the expression we arrive at the claimed result.
  \end{proof}
\end{lem}
\begin{lem}\label{lem:hafinv}
  Let $P_\pi$ be a permutation matrix. Then the following diagram commutes for any matrix $A$ and a~detection event $\bs{n}$
  \[
  \begin{tikzpicture}[baseline= (a).base]
  \node[scale=1.2] (a) at (0,0){
  \begin{tikzcd}
    A \arrow{r}{\sox\bbJ_{|\n|}} \arrow[swap]{d}{P_\pi} & A\sox\bbJ_{|\n|} \arrow{d}{\hat{P}_\pi} \\
    \tilde{A} \arrow{r}{\sox\bbJ_{|\m|}}  & \tilde{A}\sox\bbJ_{|\m|}
  \end{tikzcd}
  };
  \end{tikzpicture}
  \]
  where $\hat{P}_\pi$ is another permutation matrix and $\bs{m}=\pi(\bs{n})$.
\end{lem}
\begin{proof}
    Following the lower route, the spanning basis $\bs{\b}=(\b_1,\dots,\b_k)$ of $A$ becomes $\pi(\bs{\b})=(\b_{\pi{(1)}},\dots,\b_{\pi{(k)}})$ for $\tilde{A}$. The new basis is then expanded by considering
    \begin{equation}\label{eq:betaToalpham}
    \b_{\pi(i)}\mapsto(\a_{\pi{(i)}1},\dots,\a_{\pi{(i)}m_{\pi{(i)}}})
    \end{equation}
    and so $\pi(\bs{\a})$ is a spanning basis of $\tilde{A}\sox\bbJ_{|\m|}$. Going through the upper route, we observe that $\bs{\b}\mapsto\bs{\a}$ by the action of
    \begin{equation}\label{eq:betaToalphan}
    \b_i\mapsto(\a_{i1},\dots,\a_{in_i}).
    \end{equation}
    Then, a permutation matrix exists transforming~\eqref{eq:betaToalphan} into~\eqref{eq:betaToalpham}. Its construction is straightforward. The reordering (permutation) $(\a_{i1},\dots,\a_{in_i})\mapsto(\a_{\pi{(i)}1},\dots,\a_{\pi{(i)}n_{\pi{(i)}}})$ is followed by setting $n_{\pi{(i)}}=m_{\pi{(i)}}$. Naturally, the overall transformation is an action of a permutation matrix we denoted by $\hat{P}_\pi$.
\end{proof}
We use another result from~\cite{bradler2017gaussian} to prove the following lemma.
\begin{lem}\label{lem:determinants}
  The matrices $\s_{Q,G_i}$ of two isospectral graphs $G_1,G_2$ encoded as adjacency matrices $A_1,A_2$ of dimension $2M$ satisfy
  \[
  \det{\s_{Q,G_1}}=\det{\s_{Q,G_2}}.
  \]
\end{lem}
\begin{proof}
  In order to encode an arbitrary graph  we take two copies of a graph's adjacency matrix~\cite{bradler2017gaussian}. Also, it is advantageous to rewrite~\eqref{eq:sigmaQ} as
  \begin{align}\label{eq:sigmaAdifferent}
      \s_{A_i} &= {1\over2}(\bbI_{2M}+X_{2M}A_i)(\bbI_{2M}-X_{2M}A_i)^{-1}=
      {1\over2}(\bbI_{2M}-X_{2M}A_i)^{-1}(\bbI_{2M}+X_{2M}A_i).
  \end{align}
  The matrix $A_i$ commutes with $X_{2M}$~\cite{bradler2017gaussian} and so the eigenvalues of $X_{2M}A_i$ are products of eigenvalues of the constituents. Furthermore, using the second equality in~\eqref{eq:sigmaAdifferent} and $\s_{Q,G_i} = \s_{A_i} + \bbI_{2M}/2$ we conclude by a direct calculation that the eigenvalues of  $\s_{Q,G_1}$ and $\s_{Q,G_2}$ coincide. The claim follows from the fact that the determinant is a product of eigenvalues.
\end{proof}

\subsection{GBS and a complete set of graph invariants}\label{subsec:graphinv}

\begin{rem}\label{rem:doubling}
    We will use the transformation $A\mapsto C=A\oplus A$ for the application of GBS to the graph isomorphism problem. We recall $C\in\bbR^{2M\times2M}$. By `doubling' $A$, one copy is conveniently spanned by $(\b_1,\dots,\b_{M})$ whereas the second one by $(\ol\b_1,\dots,\ol\b_{M})$. Moreover, to make the matrix $C$ GBS-encodable (see Def.~\ref{def:GBSenc}) we simply take $R=c(C+k\bbI_{2M})=c(A\oplus A+k\bbI_{2M})$ where  $0<c<1/(\|A\|_2+k)$ for $k\geq0$. The additional multiple of an identity on the diagonal does not affect the hafnian of $A\oplus A$ as follows from Lemma~\ref{lem:hafsquared} but it will become useful in the next sections.
\end{rem}

\subsection*{GBS and Moments of Multivariate Gaussians}
The moments $\mu_{n_1,...,n_{2M}}(\Sigma)$ of a (zero-mean) $2M$-dimensional multivariate real normal distribution $\Ncal(0,\Sigma)$ are given by the following formula:
\begin{equation}
 \mu_{n_1,\dots,n_{2M}}(\Sigma)=
                \partial^{|\n|}_{\x}
		       e^{\frac{1}{2}\bs{x}^\top\Sigma \bs{x}}\big\rvert_{\bs{x}=0},
\end{equation}
where $\Sigma$ is the covariance matrix. This follows from the fact that $\exp{[\frac{1}{2}\bs{x}^\top\Sigma \bs{x}]}$ is the moment-generating function of the multivariate normal.

Let $R=c(A \op A+k\bbI_{2M})$ be GBS encodable for $k$ sufficiently high  so that $R\succ0$. Then
\begin{equation}
 p(\bs{n}) = \frac{1}{\mathbf{n}!\sqrt{\det{\sigma_Q}}}
                             \partial^{|\n|}_{\bs{\b},\ol{\bs\b}}
                             e^{\frac{1}{2}\bs{\g}^\top R \bs{\g}}\big\rvert_{\bs\g=0}.
\end{equation}
is exactly the moment $\mu_{\bs{n}}(R)$ of the $2M$-dimensional (zero-mean) multivariate normal distribution $\Ncal(0, R)$ if we ignore the prefactor $(\bs{n}!\sqrt{\det{\sigma_Q}})^{-1}$. For clarity, we have changed variables so that $x_j=\b_j$ and $x_{M+j}=\ol\b_j$ for $j=0,\dots,M$.

From the above equations, it is clear that the different photon-counting probabilities of a GBS setup are directly related to various moments of a multivariate normal distribution. Importantly, however, they do not give us \emph{all moments} $(\mu_{n_1,\dots,n_m,n_{m+1},\dots,n_{2M}}(R))$, but rather the smaller set

\noindent
$(\mu_{n_1,\dots,n_M,n_1,\dots,n_{M}}(R))$. This is something we need to be careful of.  In \cite{bradler19duality} the authors relate these moments to the matching of the prism graph induced by the graph $G$.

The moment-generating function factorizes:
\begin{equation}
 \exp{[\tfrac{1}{2}\bs{x}^\top R \bs{x}]}
 =\exp{[\tfrac{c}{2}(\bs{x}^{(M)})^\top (A+k\bbI_M)\bs{x}^{(M)}]} \times \exp{[\tfrac{c}{2}(\bs{x}^{(2M)})^\top(A+k\bbI_M)\bs{x}^{(2M)}]},
\end{equation}
where we have used the notation $\bs{x}^{(M)}\df(x_1, \dots, x_{M})$ and $\bs{x}^{(2M)}\df(x_{M+1}, \dots, x_{2M})$. We set $c=1$ (since we omit the determinant prefactor where it otherwise plays a role) and also $k=0$. This step will cost us the positive-definiteness of $A$ but at the moment this is just a formality to properly define the moment generating function. We would have set $k=0$ afterwards anyway to recover the correct probability expression. The moments of this factorized distribution are then
\begin{equation}
 \mu_{n_1,\dots,n_M,n_{M+1},\dots,n_{2M}}(A) = 
                                    \partial^{|\n|}_{\x}
                                    \big[
		                            \exp{[\tfrac{1}{2}(\bs{x}^{(M)})^\top A\,\bs{x}^{(M)}]}
                                    \times\exp{[\tfrac{1}{2}(\bs{x}^{(2M)})^\top A\,\bs{x}^{(2M)}]}
                                    \big]
		                            \big\rvert_{\bs{x}=0}
\end{equation}
Rewriting, we find
\begin{align}\label{eq:moments}
 \mu_{n_1,\dots,n_M,n_{M+1},\dots,m_{2M}}(A) & = \big[\displaystyle
                                        \partial^{|\n|}_{\x^{(M)}}
		                              \exp{[\tfrac{1}{2}(\bs{x}^{(M)})^\top A\,\bs{x}^{(M)}]}\big]\big|_{\bs{x^{(M)}}=0}
		                              \big[
                                    \partial^{|\n|}_{\x^{(2M)}}
		                              \exp{[\tfrac{1}{2}(\bs{x}^{(2M)})^\top A\,\bs{x}^{(2M)}]}\big]\big|_{\bs{x^{(2M)}}=0}\nn\\
		                          & =  \mu_{n_1,\dots,n_N}(A)\times\mu_{n_{N+1},\dots,n_{2N}}(A),
\end{align}
where $\mu_{n_1,\dots,n_M}(A)$ and $\mu_{n_{M+1},\dots,n_{2M}}(A)$ are moments of the $M$-dimensional normal distributions $\Ncal(0,A)$.

Connecting back to photon-counting probabilities, we recover $c$ and conclude that, for the considered case of block-diagonal $R$,
\begin{equation}
 p(\n) = \frac{c^{|\n|}}{\n!\sqrt{\det{\s_Q}}}\mu^2_{n_1,\dots,n_M}(A).
\end{equation}
Finally, we note that the moments are exactly the hafnian of some appropriate matrix, so
\begin{equation}\label{eq:probDoubled}
 p(\bs{n})=\frac{c^{|\n|}}{\n!\sqrt{\det{\s_Q}}}\haf{[A^{\oplus2} \sox\bbJ_{2|\n|}]}
 =\frac{c^{|\n|}}{\bs{n}!\sqrt{\det{\s_Q}}}\haf^{\,2}{[A \sox\bbJ_{|\n|}]},
\end{equation}
where the second equality also follows from~Lemma~\ref{lem:hafsquared} and~\ref{lem:linearity}.

\begin{prop}\label{prop:GIwitness}
Suppose we have two isospectral graphs $G_1$ and $G_2$. Assume we can encode the adjacency matrices $A_i$ of either graph into a Gaussian boson sampling setup.
Then these graphs are isomorphic iff the hafnians are related by a permutation, $\haf{[A_1\sox\bbJ_{|\bs{n}|}]}$ = $\haf{[A_2\sox\bbJ_{|\pi(\bs{n})|}]}$, for all $\bs{n}$. Furthermore, the permutation $\pi$ must be the same for all $\bs{n}$.
\end{prop}
\begin{proof}[Proof of $\Rightarrow$:] Suppose $G_1$ and $G_2$ are isomorphic. Equivalently, their adjacency matrices are related by a permutation
\begin{equation}\label{eq:perm}
 A_1 = P^\top A_2 P,
\end{equation}
where $P_{\pi(i)i}=\delta_{i,\pi(i)}$ for some permutation $\pi$. If we encode these adjacency matrices directly into the covariance matrices of two Gaussian states, then graph isomorphism is equivalent to the multivariate normal distributions corresponding to these two Gaussian states being related by a permutation of coordinates:
\begin{equation}
 \Ncal(0, A_1\oplus A_1)
 = \Ncal(0, (P^\top A_2 P) \oplus (P^\top A_2 P))
 = \Ncal(0, (P^{\oplus 2})^\top (A_2 \oplus A_2) (P^{\oplus 2})).
\end{equation}
All moments of these $2M$-dimensional distributions must correspondingly be related by the permutation $\pi\oplus\pi$,
\begin{equation}
 \label{eq:mu_relation}
 \mu_{n_1,\dots,n_M,n_{M+1},\dots,n_{2M}}(R_1)=\mu_{(\pi\oplus\pi)(n_1,\dots,n_M,n_{M+1},\dots,n_{2M})}(R_2),\quad\forall~(n_1,\dots,n_{2M}),
\end{equation}
where $R_i=A_i\oplus A_i$ are made into GBS encodable matrices (we keep on omitting the $c$ factors). Looking back to Eq. (\ref{eq:moments}), we have
\begin{equation}
  \mu_{n_1,\dots,n_M}(A_1)\mu_{n_{M+1},\dots,n_{2M}}(A_2)
 =\mu_{\pi(n_1,\dots,n_M)}(A_2)\mu_{\pi(n_{M+1},\dots,n_{2M})}(A_2),\quad\forall~(n_1,\dots,n_{M}), (n_{M+1},\dots,n_{2M}).
\end{equation}
These moments must be equal for any choices $\bs{p}=(n_1,...,n_{M})$ and $\bs{q}=(n_{M+1},...,n_{2M})$. Thus, we conclude that
\begin{equation}
  \label{eq:double_hafperm}
   \haf{[A_1 \sox \bbJ_{|\bs{p}|}]} \haf{[A_1 \sox \bbJ_{|\bs{q}|}]}
 = \haf{[A_2 \sox \bbJ_{|\pi(\bs{p})|}]} \haf{[A_2 \sox\bbJ_{|\pi(\bs{q})|}]}.
\end{equation}
In particular, for $\bs{p}=\bs{q}=\bs{n}$, where $\bs{n}$ is arbitrary, we get $\haf^{\,2}{[A_1 \sox \bbJ_{|\bs{n}|}]} = \haf^{\,2}{[A_2 \sox \bbJ_{|\pi(\bs{n})|}]}$. We now use the fact that adjacency matrices $A$ contain only $0$s or $1$s, so $\haf{[A\sox\bbJ_{|\bs{n}|}]}\geq 0$ for any possible $A$ or $\bs{n}$.
This leads to
\begin{equation}\label{eq:hafperm}
   \haf{[A_1 \sox \bbJ_{|\bs{n}|}]}
 = \haf{[A_2 \sox \bbJ_{|\pi(\bs{n})|}]}, ~ \forall ~ \bs{n},
\end{equation}
which proves the statement.

\paragraph{Proof of $\Leftarrow$:} Eq. (\ref{eq:hafperm}) immediately implies Eq. (\ref{eq:double_hafperm}), even when the hafnians are not positive. Furthermore, Eqns. (\ref{eq:perm})-(\ref{eq:double_hafperm}) are all equivalent. Hence, the graphs having adjacency matrices $A_1$ and $A_2$ are isomorphic.
\end{proof}

\subsection*{GBS and Symmetrized  Moments of Multivariate Gaussians}

In this part we find a new criterion for isomorphism of two isospectral graphs by showing that symmetrized moments are also complete invariants for graph isomorphism:
\begin{thm}\label{thm:mainthm}
    Let $G_1$ and $G_2$ be two isospectral graphs on an even number of vertices $M$.  Denote by $p_1(\bs{n})$ and $p_2(\bs{n})$, the probabilities corresponding to $G_1$ and $G_2$, given in Eq.~\eqref{eq:ProbMixedGBS}.  Then $G_1$ and $G_2$ are isomorphic if and only if the symmetrized sums
    \[
    \sum_{\sigma\in \mathfrak{S}_n} \sqrt{p(\bs{n}_{\sigma})}
    \]
    are the same for the two graphs for all possible $\bs{n}$.
\end{thm}
To avoid a notational clash in this section we use $\n_\s$ instead of $\s(\n)$ used in the previous sections.
We start with the following result.
\begin{thm}\label{thm:isogausdist}  The following statements are equivalent for two Gaussian distributions  with zero mean and positive definite covariance matrices $\Sigma,\Sigma'\in \bbR^{n\times n}$:
\begin{enumerate}
\item The two Gaussian distributions are isomorphic.
\item The matrices $\Sigma$ and $\Sigma'$ are permutationally similar.
\item The matrices $\Sigma^{-1}$ and $(\Sigma')^{-1}$ are permutationally similar.
\item For each homogeneous symmetric polynomial $p(\bs{x})$ of even degree  the expected value of $p(\bs{x})$ is the same for the two Gaussian distributions.
\item The symmetrized moments of the two Gaussian distributions are the same.
\end{enumerate}
\end{thm}
\begin{proof}
  For a given Gaussian distribution with zero mean and the covariance matrix $\Sigma$ let us denote by $H_{\Sigma}$ the distribution with the following density
    \begin{equation}\label{eq:symgausdist}
    h_{\Sigma}(\bs{x})=\frac{1}{n!}\sum_{\sigma\in\mathfrak{S}_n} \exp{[-\bs{x}^\top P^\top(\sigma) \Sigma^{-1}P(\sigma)\bs{x}]}.
    \end{equation}
    By $G_{\Sigma}$ we denote the density function of the normal distribution $\exp{[-\x^\top \Sigma^{-1}\x]}$.

    Clearly if $\Sigma$ and $\Sigma'$ are permutationally similar then $h_{\Sigma}(\bs{x})=h_{\Sigma'}(\bs{x})$ for each $\bs{x}$.  So we need to prove the other direction.
    Let $p(\bs{x}) $ be a monomial $x_1^{m_1}\cdots x_n^{m_n}$ of even degree.  We denote $\bs{m}=(m_1,\ldots,m_n)$.
    Consider the moments
    \begin{align*}
    \mu_{G_\Sigma}(\bs{m})&=\bbE_{G_{\Sigma}}[X_1^{m_1}\cdots X_n^{m_n}],\\
    \mu_{H_\Sigma}(\bs{m})&=\bbE_{H_{\Sigma}}[X_1^{m_1}\cdots X_n^{m_n}].
    \end{align*}
    Then
    \[
    \mu_{H_\Sigma}(\bs{m})=\frac{1}{n!} \sum_{\sigma\in\mathfrak{S}_n}\mu_{G_\Sigma}(\bs{m}_{\sigma}).
    \]
    We call $\mu_{H_\Sigma}(\bs{m})$ the symmetrized moments of $G_{\Sigma}$.  Clearly, $\mu_{H_\Sigma}(\bs{m})=\mu_{H_\Sigma}(\bs{m}_{\sigma})$ for each $\sigma\in\mathfrak{S}_n$. Thus if two Gaussian distributions are isomorphic, then $H_{\Sigma}=H_{\Sigma'}$ and the corresponding density functions are the same.  In particular, the moments of $H_{\Sigma}$ and $H_{\Sigma'}$ are the same.

    Our first main result is the claim that if the moments of $H_{\Sigma}$ and $H_{\Sigma'}$ are the same then $H_{\Sigma}=H_{\Sigma'}$.  It is not true that the equality of the moments yield that the distribution are the same.  However, it is true for  all distributions that have the form of $H_{\Sigma}$.  Indeed a sufficient condition is
    $M(\bs{u}) = \bbE[\exp{[\langle \bs{u},\bs{X}\rangle]}]$ is a well-defined vector for all $\bs{u}\in\bbR^n$~\cite{kleiber2013multivariate}.  This condition is satisfied for $H_{\Sigma}$, where $\Sigma$ is a positive definite matrix.

    It is left to show that if $H_{\Sigma}=H_{\Sigma'}$  then $\Sigma$ and $\Sigma'$ are permutationally similar.  We first analyze the behavior of the $n!$ numbers
    \begin{equation}\label{eq:allexpval}
    \exp{[-\bs{x}^\top P^\top(\sigma) \Sigma^{-1}P(\sigma)\bs{x}]},
    \end{equation}
          for a fixed but arbitrary vector $\bs{x}\in\bbR^n$ and   $\sigma\in\mathfrak{S}_n$.

    Let $\orb{\Sigma}$ be all pairwise distinct matrices of the form $P^\top(\sigma) \Sigma P(\sigma)$ for $\sigma\in\mathfrak{S}_n$.  Recall that $|\orb{\Sigma}|$ divides $n!$ and $n!/|\orb{\Sigma}|$ is the cardinality of the automorphism group, i.e., all $\sigma\in\mathfrak{S}_n$ such that $P^\top(\sigma) \Sigma P(\sigma)=\Sigma$. Let $A_1,A_2$ be two real symmetric matrices of order $n$.  Define two corresponding quadratic forms $f_1(\x)= \x^T A_1\x, f_2(\x)=\x^\top A_2\x$.  Assume that $A_1\ne A_2$. Then $f_1(\x)=f_2(\x)$ if and only if $h(\x)=\x^\top (A_1-A_2)\x=0$.
    Hence for generic, (randomly selected $\x$) we have that $f_1(\x)\ne_2(\x)$.  Similarly: let $A_1,\ldots,A_k$ be $k$ pairwise distinct
    symmetric matrices.  Set $f_i(\x)=\x^\top A_i\x$ for $i\in[k]$.  Then for generic $\x$, $f_i(\x)\ne f_j(\x)$ for $i\ne j$. Since any two pairs of matrices in $\orb{\Sigma}$ are pairwise distinct it follows that for a generic $\bs{x}$ we will have exactly $|\orb{\Sigma}|$ distinct values in~\eqref{eq:symgausdist} and each value is repeated $n!/|\orb{\Sigma}|$ times.

    Assume now that  for each $\bs{x}\in\bbR^n$ in~\eqref{eq:symgausdist}  we have the equality $h_{\Sigma}(t\bs{x})=h_{\Sigma'}(t\bs{x})$:
    \begin{equation}\label{eq:symgausdisteq}
    \sum_{\sigma\in\mathfrak{S}_n} \big(\exp{[-\bs{x}^\top P^\top(\sigma)\Sigma^{-1}P(\sigma)\bs{x}]}\big)^{t^2}=\sum_{\sigma\in\mathfrak{S}_n} \big(\exp{[-\bs{x}^\top P(\sigma)^\top (\Sigma')^{-1}P(\sigma)\bs{x}]}\big)^{t^2}
    \end{equation}
    for some fixed $t\ge 0$.  Fix $\bs{x}$ in general position. Then for~$t=1$ we get that the number of distinct values in \eqref{eq:symgausdist} is $|\orb{\Sigma}|$ for $\Sigma$ and $|\orb{\Sigma'}|$ for $\Sigma'$, respectively.
    Let
    \begin{align}
      a(\sigma)  &= \exp{[-\bs{x}^\top P(\sigma)^\top (\Sigma)^{-1}P(\sigma)\bs{x}]}, \\
      a'(\sigma) &= \exp{[-\bs{x}^\top P(\sigma)^\top (\Sigma')^{-1}P(\sigma)\bs{x}]}
    \end{align}
    for $\sigma\in \mathfrak{S}_n$. In the equality \eqref{eq:symgausdisteq} set $t=\sqrt{k}$ for ${k=0,1,\ldots,n!}$.  Thus we have the equalities:
      \[
      \sum_{\sigma\in \mathfrak{S}_n} a(\sigma)^k=\sum_{\sigma\in \mathfrak{S}_n} a'(\sigma)^k
      \]
    for $k=1,\ldots,n!$. These equalities yield that the two multisets $\{a(\sigma),\sigma\in \mathfrak{S}_n\}$ and $\{a'(\sigma),\sigma\in \mathfrak{S}_n\}$ are the same. Hence the $n!$ moments of discrete distributions equally distributed on $n!$ points given in  \eqref{eq:allexpval} for $\Sigma$ and $\Sigma'$ are the same.  Hence these two multisets are the same.  First it yields that  $|\orb{\Sigma}|=|\orb{\Sigma'}|$.  Moreover there exists $P(\sigma)$ such that $\bs{x}^\top P^\top(\sigma)\Sigma^{-1}P(\sigma)\bs{x}=\bs{x}^\top (\Sigma')^{-1}\bs{x}$.  Moreover, for each $\Sigma_i=P_i^\top \Sigma P_i$ in the orbit of $\Sigma$ (under the action of the group of permutations) we have a permutation $Q_i$ such that  $\bs{x}^\top Q_i^\top(\Sigma_i')^{-1}Q_i\bs{x}=\bs{x}^\top\Sigma_i^{-1}\bs{x}$. Now if we change $\bs{x}$ to $\bs{y}$ we still have the same equality $\bs{y}^\top P^\top(\sigma)\Sigma^{-1}P(\sigma)\bs{y}=\bs{y}^\top (\Sigma')^{-1}\bs{y}$. This finally shows that $P^\top(\sigma)\Sigma^{-1}P(\sigma)=\Sigma'$.  So indeed the covariance matrices are permutationally similar.
\end{proof}

Let $\be_i=(\delta_{i,1},\ldots,\delta_{i,M})$, where $\delta_{i,j}$ is the Kronecker delta function. Assume that $\n=(n_1,\ldots,n_M)\in\bbZ_+^M$. Let $B=[b_{i,j}]$ be a real symmetric matrix and recall  the ``$\n$-th  moment corresponding to $B$'' from the beginning of this section as
\begin{eqnarray}\label{eq:formdefmom}
    \mu(\n,B)= \frac{\partial^{|\n|}}{\partial \x^{\n}}
    \exp{\big[\frac{1}{2}\x^\top B\x\big]}\Big\rvert_{\x=0}=\frac{1}{(|\n|/2)!}\frac{\partial^{|\n|}}{\partial \x^{\n}}
    \big(\frac{1}{2}\x^\top B\x\big)^{|\n|/2}\Big\rvert_{\x=0}.
\end{eqnarray}
(If $B\succ 0$ then $E(X_1^{n_1}\cdots X_M^{n_M})$, the $\n$-th moment of the Gaussian distribution given by the covariance $B$, is equal to $\mu(\n,B)$ up to a multiplicative constant.)

To proceed, we also recall the generalization of the classical Leibniz's formula of the derivative of the product of $m$ functions in one variable:
\[
    (\prod_{i=1}^m f_i)^{(n)}=\sum_{a_1,\ldots,a_m\in\bbZ_+, \sum_{i=1}^m a_i=n}  {n\choose a_1,a_2,\ldots,a_m}\prod_{i=1}^m f_i^{(a_i)}, \quad {n\choose a_1,a_2,\ldots,a_m}=\frac{n!}{a_1!\times\ldots\times a_m!}.
\]
Assume now that $f_1=\cdots=f_m=f(\x)=f(x_1,\ldots,x_M)$.  Then for $\n=(n_1,\ldots,n_M)\in\bbZ_+^M$ we denote by $\partial^{\n}=\partial^{n_1}_1\cdots\partial^{n_M}_M$.  For $\n,\ba\in\bbZ_+^M$  let ${\n\choose\ba}=\prod_{i=1}^M {n_i \choose a_i}$.  Then Leibniz's formula yields the multilinear Leibniz's formula:
\begin{align}\label{eq:MLF}
 \partial^{\n}(f^m)&=\sum_{\sum_{i=1}^m \ba_i=\n} {\n\choose \ba_1,\ba_2,\ldots,\ba_m}\prod_{i=1}^m (\partial^{\ba_i}f),
\end{align}
where
\begin{equation}
  {\n\choose \ba_1,\ba_2,\ldots,\ba_m}=\prod_{j=1}^M {n_j\choose a_{1,j},a_{2,j},\ldots,a_{m,j}}
\end{equation}
and $\ba_i=(a_{i,1},\ldots,a_{i,M})$ and $i\in[m]$.
We now apply this formula for $f(\x)=\frac{1}{2}\x^\top B\x$ and $m=|\n|/2$. Then in \eqref{eq:MLF} we need to consider only the case where $|\ba_i|=2$ for each $i\in|\n|/2$. So
\begin{align}\label{eq:expformuB}
 \mu(\n,B)=\frac{1}{(|\n|/2)!}\sum_{\sum_{i=1}^{|\n|/2} \ba_i=\n, |\ba_1|=\ldots=|\ba_{|\n|/2}|=2} {\n\choose \ba_1,\ba_2,\ldots,\ba_{|\n|/2}}\prod_{i=1}^{|\n|/2} \partial^{\ba_i}(1/2 (\x^\top B\x)).
\end{align}
We have two kinds of $\ba_i$.  Namely, either $\ba_i=2\be_p$ or $\ba_i=\be_p+\be_q$, where $1\le p<q\le n$. Let us discuss briefly all possibilities for the decomposition of $\n$ as $\n=\sum_{i=1}^{|\n|/2}\ba_i$.  We claim that the set $\{\ba_1,\ldots,\ba_{|\n|/2}\}$ corresponds to the following multigraph $G=G(\ba_1,\ldots,\ba_{|\n|/2})$ with multiple edges and self loops.  Each $\ba_i=\be_p+\be_q$, where $1\le p <q\le n$ corresponds to an edge $\{p,q\}$.  Each $\ba_i=2\be_p$ corresponds to a self loop on vertex~$p$ (the degree of a self loop is $2$). So  $A(G)=[c_{pq}(G)]$, the adjacency matrix of $G$, is a symmetric matrix whose entries are nonnegative integers with the following properties.  Each diagonal entry $c_{pp}(G)$ is an even integer.  $c_{pp}(G)/2$ is  the number of $\ba_i$ of the form $2\be_p$.  For $1\le p<q\le M$ the integer $c_{pq}(G)$ is the number of $\ba_i$ of the form $\be_p+\be_q$.

Let $2k=\sum_{p=1}^M c_{pp}$ be a nonnegative integer.  That is, the set $\{\ba_1,\ldots,\ba_{|\n|/2}\}$ has $k$ vectors of the form $2\be_p$ for all possible $p\in[n]$. Assume that $k=0$.  Then $\{\ba_1,\ldots,\ba_{|\n|/2}\}$ correspond to a given multigraph $G=G(\ba_1,\ldots,\ba_{|\n|/2})$ with no loops.  If one permutes the vectors $\ba_1,\ldots,\ba_{|\n|/2}$ one obtains the same loopless multigraph $G$, whose degree sequence is $\n=(n_1,\ldots,n_M)$, where $n_i=\sum_{p=1}^M c_{ip}$.  Let $\Gcal(\n)$ be all loopless multigraphs $G$ whose degree sequence is $\n$.  We arrange the edges of $G$ in a fixed (say lexicographic order): $(1,2),\ldots,(1,M), (2,3),\ldots,(M-1,M)$. For example the sequence of edges on $3$ vertices $(2,3), (1,3), (1,2), (1,3)$ is arranged as $(1,2), (1,3), (1,3),(2,3)$.  It corresponds to a degree sequence $\n=(3,2,3)$.

We denote by $\rS_{M,0}$ all $M\times M$ symmetric matrices with zero diagonal.  Assume that $A\in \rS_{M,0}$.  Note that for $f=1/2 (\x^\top A\x)$ we get that $\partial^2_i f=0$ for each $i$.
\begin{lem}\label{lem:formudA}
 Let $A\in\rS_{M,0}$.  Then
 \begin{align}\label{eq:mudAfor}
 \mu(\n,A)=\sum_{G(\ba_1,\ldots,\ba_{|\n|/2})\in\Gcal(\n)} \frac{1}{\prod_{1\le p<q\le n} c_{pq}(G(\ba_1,\ldots,\ba_{|\n|/2}))!} {\n\choose \ba_1,\ba_2,\ldots,\ba_{|\n|/2}}\prod_{i=1}^{|\n|/2} (\partial^{\ba_i}f).
 \end{align}
\end{lem}
\begin{proof}
    Given a decomposition $\sum_{i=1}^{|\n|/2}\ba_i=\n$, corresponding to the graph $G(\ba_1,\ldots,\ba_{|\n|/2})$,  how may different decompositions are there?  Since the edge $\{p,q\}, 1\le p<q\le M$ appears $c_{pq}(G(\ba_1,\ldots,\ba_{|\n|/2}))$ times, the number of different decompositions is $\frac{(|\n|/2)!}{\prod_{1\le p<q\le n} c_{pq}(G(\ba_1,\ldots,\ba_{|\n|/2}))!}$.  Use \eqref{eq:expformuB} to deduce \eqref{eq:mudAfor}.
\end{proof}
For a given $k\in[|\n|/2]$ denote by
\[
    \Fcal_k(\n)=\big\{(j_1,\ldots,j_k)\in\bbN^k, 1\le j_1\le\cdots\le j_k\le M,\ 2\sum_{l=1}^k \be_{j_l}\le \n\big\}.
\]
For each $(j_1,\ldots,j_k)\in\Fcal_k(\n)$ and $i\in[M]$ denote $m_i(j_1,\ldots,j_k)$ the number of $j_l$ that are equal to $i$. So $\sum_{i=1}^M m_i(j_1,\ldots,j_k)=k$.
\begin{lem}\label{lem:momrect}
    Let $A\in\rS_{M,0}$ and $t\in\bbR$ are given. Assume that $|\n|$ is even.  Then
    \begin{align}\label{eq:mubdbfor}
     \mu(\n,t\bbI_M+A)=\mu(\n,A)\\
     +\sum_{k=1}^{|\n|/2} t^k\sum_{(j_1,\ldots,j_k)\in\Fcal_k(\n)}\Big(\prod_{i=1}^M\frac{d_i!}{m_i(j_1,\ldots,j_k)!2^{m_i(j_1,\ldots,j_k)}(d_i-2m_i(j_1,\ldots,j_k))!}\Big)
     \mu(\n-2\sum_{l=1}^k \be_{j_l},A).\nn
    \end{align}
    Here $\mu(0,A)=1$.
\end{lem}
\begin{proof}
    We consider the formula~\eqref{eq:expformuB}.  Suppose we have $\{\ba_1,\ldots,\ba_{|\n|/2}\}$ satisfying:  (i) $|\ba_l|=2$ for each~$l$, and, (ii) $\sum_{l=1}^{|\n|/2}\ba_l=\n$.  These terms define $G=G(\ba_1,\ldots,\ba_{|\n|/2})$.  In how may different ways  can we represent $\n$ corresponding to $G(\ba_1,\ldots,\ba_{|\n|/2})$?  We can do it by permuting the factors $\ba_1,\ldots,\ba_{|\n|/2}$.  As in the proof of Lemma~\ref{lem:formudA} it is
    \[
        \frac{(|\n|/2)!}{\prod_{1\le p<q\le M} c_{pq}(G(\ba_1,\ldots,\ba_{|\n|/2}))!\prod_{p=1}^M (c_{pp}(G(\ba_1,\ldots,\ba_{|\n|/2}))/2)!}.
    \]
    The entry $ c_{pq}(G(\ba_1,\ldots,\ba_{|\n|/2}))$ stands for the number of times the edge $\{p,q\}$ appears.  The number of selfloops $(p,p)$ is $c_{pp}(G(\ba_1,\ldots,\ba_{|\n|/2}))/2)$.  Dividing by $(|\n|/2)!$ we see that the contribution of $G(\ba_1,\ldots,\ba_{|\n|/2})$ is
    \begin{eqnarray*}
        \frac{1}{\prod_{1\le p<q\le M} c_{pq}(G(\ba_1,\ldots,\ba_{|\n|/2}))!\prod_{p=1}^M (c_{pp}(G(\ba_1,\ldots,\ba_{|\n|/2}))/2)!}\\
        \times{\n\choose \ba_1,\ba_2,\ldots,\ba_{|\n|/2}}\prod_{i=1}^{|\n|/2} \partial^{\ba_i}(1/2 (\x^\top B\x)).
    \end{eqnarray*}
    Let $k$ be the number of terms in $\{\ba_1,\ldots,\ba_{|\n|/2}\}$ of the form $2\be_i$.  The contribution of all terms $\{\ba_1,\ldots,\ba_{|\n|/2}\}$ for which $k=0$ is $\mu(\n,A)$.  Let us assume that $k\in [|\n|/2]$.  Without loss of generality we can assume that $\ba_l=2\be_{j_l}$ for $l\in[k]$.  So $\ba_l=\be_{p(l)} +\be_{q(l)}$, $p(l)<q(l)$ for $l>k$.  Clearly, $\partial^{\ba_l} f=t$ for $l\in[k]$.  Hence $\prod_{l=1}^k \partial^{\ba_l} f=t^k$.  The sum of all
    \[
    \frac{1}{\prod_{1\le p<q\le M} c_{pq}(G(\ba_{k+1},\ldots,\ba_{|\n|/2}))!}
    {\n-2\sum_{l=1}^k \be_{j_l}\choose \ba_{k+1},\ldots,\ba_{|\n|/2}}\prod_{l=k+1}^{|\n|/2} \partial^{\ba_l} f,
    \]
    where $\ba_l=\be_{p(l)} +\be_{q(l)}$ for $l>k$, and $\n-2\sum_{l=1}^k \be_{j_l}=\sum_{l=k+1}^{|\n|/2} \ba_l$ is exactly $\mu(\n-2\sum_{l=1}^k \be_{j_l},A)$.  It is left to justify the coefficient $\prod_{i=1}^M\frac{n_i!}{m_i(j_1,\ldots,j_k)!2^{m_i(j_1,\ldots,j_k)}(n_i-2m_i(j_1,\ldots,j_k))!}$ in front of $\mu(\n-2\sum_{l=1}^k \be_{j_l},A)$.
    This comes from the equality
    \begin{align*}
    &\frac{1}{\prod_{1\le p<q\le M} c_{pq}(G(\ba_{1},\ldots,\ba_{|\n|/2}))!\prod_{p=1}^M (c_{pp}(G(\ba_1,\ldots,\ba_{|\n|/2}))/2)!}\\
    &\times{\n\choose \ba_1,\ba_2,\ldots,\ba_{|\n|/2}}\\
    &=
    \Big(\prod_{i=1}^M\frac{n_i!}{m_i(j_1,\ldots,j_k)!2^{m_i(j_1,\ldots,j_k)}(n_i-2m_i(j_1,\ldots,j_k))!}\Big)\frac{1}{\prod_{1\le p<q\le M} c_{pq}(G(\ba_{k+1},\ldots,\ba_{|\n|/2}))!}\\
    &\times{\n-2\sum_{l=1}^k \be_{j_l}\choose \ba_{k+1},\ldots,\ba_{|\n|/2}}.
    \end{align*}
    We need to see this equality on the level of the derivative with respect to the variable $x_i$.  Let $m_i=m_i(j_1,\ldots,j_k)$.  If $m_i=0$, then $(\n-2\sum_{j=1}^k \be{_{j_l}})_i=n_i$ for the coordinate $i$ we have obvious equality. Assume now that $m_i\ge 1$. Then on the left-hand side of the above equality we the factor $\frac{n_i!}{(2!)^{m_i}m_i!}$.  The factor $m_i!$ is equal to $(c_{ii}(G)/2)!$.  On the right-hand side we have the factors:
    \[
    \frac{n_i!}{m_i!(2!)^{m_i}(n_i-2m_i)!}(n_i-2m_i)!.
    \]
\end{proof}


We are now ready to give the proof for the main result of this section.
\begin{proof}[Proof of Theorem~\ref{thm:mainthm}.]
    We write the symmetrized sums of  $\mu(\n,A)$ as
    \begin{equation}\label{eq:symmu}
      \mu_{\mathrm{sym}}(\n,A)=\sum_{\s\in\mathfrak{S}_n}\mu(\n,P^\top(\s)AP(\s)).
    \end{equation}
    Now assume $\mu_{\mathrm{sym}}(\n,A_1)=\mu_{\mathrm{sym}}(\n,A_2),\forall\n$ for two isospectral graphs $A_1,A_2$. Then from~\eqref{eq:mubdbfor} we get $\mu_{\mathrm{sym}}(\n,B_1)=\mu_{\mathrm{sym}}(\n,B_2),\forall\n$. Finally, Theorem~\ref{thm:isogausdist} yields that $B_1$ and $B_2$ are permutationally similar and so are $A_1$ and $A_2$.
\end{proof}
\begin{exa}
  Let us illustrate Lemma~\ref{lem:momrect} and Theorem~\ref{thm:mainthm} on an example of a graph whose adjacency matrix is $A\oplus A$ of size $2M=6$ for orbit of $\n=(2,3,3)$. Since $|\n|=8$ it is clear that only the fourth power ($|n|/2=4$)  of $(\x^{(M)})^\top A(t)\x^{(M)}$ survives an encounter with the partial derivatives. So, following Lemma~\ref{lem:momrect}, we write
  \begin{subequations}
  \begin{align}\label{eq:ExampleMainThm1}
    {1\over2^44!}{\partial^8\over\partial x_1^2\partial x_2^3\partial x_3^3}((\x^{(M)})^\top A(t)\x^{(M)})^4
     = {1\over2^44!}{\partial^8\over\partial x_1^2\partial x_2^3\partial x_3^3}\sum_{k=0}^4\binom{4}{k} t^k|\x^{(M)}|^{2k}((\x^{(M)})^\top A\x^{(M)})^{4-k} \\
     \stackrel{\x=0}{\mapsto} 36a_{12}a_{13}a^2_{23}+6ta_{23}(3(a_{12}^2+a_{13}^2)+a_{23}^2)+18t^2a_{12}a_{13}+9t^3a_{23}.\label{eq:ExampleMainThm2}
  \end{align}
  \end{subequations}
  Each $t^k$ coefficient corresponds to a polynomial of the matrix entries in the exponential of some $\mu(\n-2\sum_{l=1}^k \be_{j_l},A)\equiv\mu(\m,A)$. In accordance with Eq.~\eqref{eq:mubdbfor} we get, for example for $t^2$, $\m=(2,1,1)$ since
  \[
  {1\over16}{\partial^4\over\partial x_1^2\partial x_2\partial x_3}\big((\x^{(M)})^\top A\x^{(M)}\big)^2\Big\rvert_{\x=0}=a_{12}a_{13}.
  \]
  As the final step, we symmetrize the orbit represented by $\n$ (in this case  the orbit size equals $3$) which causes a permutation of indices in~\eqref{eq:ExampleMainThm2}.
\end{exa}

It is advantageous to stratify the measurement events of an $M$-mode interferometer according to the total photon number $|\n|\geq0$. Once $M$ and $|\n|$ are fixed, all possible detection events can be split into the orbits $O_i$ (equivalence classes under permutation) that partition the set of all events for a fixed $M$ and $|\n|$. We choose the class representative to be a detection event $\n=(n_j)_{j=1}^M$ such that $n_i\leq n_j,\forall i,j$ and denote by $G_{\n}$ its stabilizer. Clearly $G_{\n}\subset G=\mathfrak{S}_M$ and the orbits are generated by the left action of the coset $G/G_{\n}$. In order to find the orbits with a great number of detection events (presumably the most likely ones) we count the orbit size according to $|O_{\n}|=|\mathfrak{S}_M|/|G_{\n}|=\binom{M}{k_0,k_1,\dots,k_\ell}$, where $k_j$ are the multiplicities of the $j$-th photon events satisfying $\sum_{j=0}^{\ell}jk_j=|\n|$ and $\ell\leq M$.  The probability of measurement of a given pattern  $(n_1,\dots,n_M)$ is given by $p(\n)$ in Eq.~\eqref{eq:ProbMixed} (or, more precisely, by its doubled version, Eq.~\eqref{eq:probDoubled}, see Lemma~\ref{lem:hafsquared} and the remark on page~\pageref{rem:doubling}), where $\sum_{i=1}^Mn_i=|\n|$. Hence the  probability of orbit $O_{\n}$ for a graph $G$ reads
\begin{equation}\label{eq:probOrbit}
  p_G(O_{\n})={1\over\sqrt{\det{\s_{Q,G}}}}{c^{|\n|}\over\bs{n}!}
  \sum_{\n\in O_{\n}}^{|O_{\n}|}\haf^{\,2}{[A\sox\bbJ_{|\n|}]}.
\end{equation}
How does the number of orbits increase with $|\n|$? This is equivalent to the question of integer $|\n|$ partition, that is, in how many ways one can write
\begin{equation}\label{eq:partition}
\la_1+\dots+\la_m=|\n|,
\end{equation}
where the order of the sum plays no role and the number of parts is $1\leq m\leq M$. We naturally order the parts such that $\la_i\leq\la_{i+1}, \forall i$. Suppose $M\geq|\n|$ first. No closed formula is known but the generating function for integer partition provides the number of orbits for a given $|\n|$. Also, very precise estimates have been uncovered and the growth of the number of orbits is exponential in $|\n|$. For $M<|\n|$, not all number partitions can be realized and the counting is given by the generating function for the number of integer partitions into exactly $M$ parts. Note that we only partition even numbers in this paper, since GBS assigns zero probability for odd $|\n|$.
\begin{cor}[of Lemma~\ref{lem:nTOT}]\label{cor:ofnTOTlemma}
  $p_G(O_{\n})=0$ whenever  $p(\n)=0$ for the orbit representative $\n$.
\end{cor}
Curiously, if we try to coarse-grain the probability distribution further and introduce the \emph{partition probability}
\begin{equation}\label{eq:probCoarsegrained}
  p_G(|\n|)\df\sum_{\n~\mathrm{s.t.\ }|\n|\mathrm{\ fixed}}p(\n)
  ={1\over\sqrt{\det{\s_{Q,G}}}}
  \sum_{n_1+\dots+n_M=|\n|}{1\over\n!}\sum_{\n\in O_{\n}}^{|O_{\n}|}\haf^{\,2}{[A\sox\bbJ_{|\n|}]},
\end{equation}
where the first sum on the RHS is over the partitions of $|\n|$ and the second sum over the orbit elements. We find
\begin{lem}\label{lem:probCoarseGr}
  $p_{G_1}(|\n|)=p_{G_2}(|\n|)$ for all $|\n|$ whenever the graphs $G_1,G_2$ are isospectral.
\end{lem}
\begin{proof}
  Any undirected graph $G$ on $2M$ vertices can be encoded as a pure covariance matrix whose circuit decomposition consists of an array of $M$ single-mode squeezing transformation $S(r_k)$ ($0\leq k\leq M$) followed by an $M$-mode linear interferometer $U$~\cite{bradler2017gaussian}. For each $S(r_k)$ we find
  \begin{equation}
    S(r_k)\ket{0}={1\over\sqrt{\cosh{r_k}}}\sum_{n=0}^{\infty}{\sqrt{(2n)!}\over2^nn!}\tanh^n{r_k}\ket{2n}
  \end{equation}
  and so
  \begin{equation}\label{eq:Sproduct}
    \bigotimes_{k=1}^M S(r_k)\ket{0}=\sum_{\n/2=0}^{\infty}\sum_{i=1}^{\binom{|\n|/2+M-1}{|\n|/2}}\a_{i\n}(r_1,\dots,r_M)\ket{(\n,M)}_i,
  \end{equation}
   where $\ket{(\n,M)}$  carry all completely symmetric representations of $su(M)$ (each representation labeled by $\n/2$). Given $\la_k(A)$ and $0<c<1/\|A\|_2$ for $G$'s adjacency matrix $A$ (\cite{bradler2017gaussian}, see also Lemma~\ref{lem:posdefsigA}) we can write $c\la_k=\tanh{r_k}$. Therefore $\a_{i\n}(r_1,\dots,r_M)=\a_{i\n}(\la_1,\dots,\la_M)$. Since the interferometer $U$ preserves the number of particles $|\n|$ the partition probability $p_{G}(|\n|)$ is unaffected by it. Then
   $$
   p_{G}(|\n|)=\sum_{i=1}^{\binom{|\n|/2+M-1}{|\n|/2}}|\a_{i\n}(r_1,\dots,r_M)|^2
   =\sum_{i=1}^{\binom{|\n|/2+M-1}{|\n|/2}}|\a_{i\n}(\la_1,\dots,\la_M)|^2.
   $$
   However, the RHS is independent on the graph (depends only on $\la_k$ common for isospectral graphs) and the claim follows.
\end{proof}
\begin{rem}
  Note that the similar argument does not hold for the less coarse-grained probability $p_G(O_{\n})$ in Eq.~\eqref{eq:probOrbit} since the interferometer `mixes' the orbits.
\end{rem}
Even though the coarse-grained probability, Eq.~\eqref{eq:probCoarsegrained}, cannot be used to distinguish nonisomorphic graphs, not all hope is lost. Possible strategies and the closely related problem of scalability is discussed in Sec.~\ref{sec:discussion}.

\subsection{Modifying the results for $C=A\oplus A$ and beyond}\label{subsec:hierarchyTower}
 Given $A$ of even order consider $p(\n,C)$ in \eqref{eq:ProbMixedGBS} and $\mu(\n,C)$. If for the two graphs $A$ and $B$ we have the equalities for the symmetrized sums
 \[
 \sum_{\sigma\in \mathfrak{S}_n} p(\n_{\sigma},A\oplus A)=\sum_{\sigma\in \mathfrak{S}_n} p(\n_{\sigma},B\oplus B),
 \]
 what can we say? If instead of considering just the matrix $A$ we will consider the matrix $A\oplus A$ then we can conclude that our bigger graph is a disjoint  union of two isomorphic graphs.  So if the union of two isomorphic graphs is isomorphic to the union of another two isomorphic graphs, then the two graphs are also isomorphic. If we consider the functions  $\mu(\m,A\oplus A)$.  Note that $\m=(m_1,\ldots,m_{2M})=(\n,\n')$ where $\n=(n_1,\ldots,n_{M}), \n'=(n_{M+1},\ldots,n_{2M})$.  It now follows that
 \[
 p((\n,\n'), A\oplus A)=p(\n,A)p(\n',A).
 \]
 Hence, if $\gamma=(\x,\y,\ol\x,\ol\y)$ then
\begin{align*}
   & \exp{\big[{1\over2}\g^\top (C\oplus C)\gamma\big]} \\
   & =\exp{\big[{1\over2}(\x^{(M)})^\top A\x^{(M)}\big]}\exp{\big[{1\over2}(\y^{(M)})^\top A\y^{(M)}\big]}
         \exp{\big[{1\over2} (\ol\x^{(M)})^\top A\ol\x^{(M)}}\big]\exp{\big[{1\over2}(\ol\y^{(M)})^\top A\ol\y^{(M)}\big]}.
\end{align*}
Therefore, if we have equalities for the symmetrized sums, we get a whole hierarchy of necessary conditions by considering $A^{\oplus k}$ for $k>2$.


\subsection{Partition-averaged photon distribution as a necessary condition for graph isomorphism}\label{subsec:meanphoton}
The main result of this section will be a simpler necessary condition for isospectral graphs to be isomorphic.
\begin{defi}\label{def:partAvPhDistro}
  Let $A$ be the adjacency matrix of an $M$-vertex graph $G$ and $1\leq k\leq M$. Then we call the partition-averaged photon distribution of the $k$-th detector the  function
    \begin{equation}\label{eq:averPhNumberfiner}
      \lan{n_k}\ran_G
      \df\sum_{\n\in O_{\n}}^{|O_{\n}|}n_kp(\n)={1\over\sqrt{\det{\s_{Q,G}}}}
      {1\over\n!}\sum_{\n\in O_{\n}}^{|O_{\n}|}n_k\haf^{\,2}{[A\sox\bbJ_{|\n|}]}
    \end{equation}
    and its coarse-grained version reads
    \begin{equation}\label{eq:averPhNumber}
      \langle\langle{n_k}\rangle\rangle_G
      ={1\over\sqrt{\det{\s_{Q,G}}}}\sum_{n_1+\dots+n_M=|\n|}{1\over\n!}\sum_{\n\in O_{\n}}^{|O_{\n}|}n_k\haf^{\,2}{[A\sox\bbJ_{|\n|}]}.
    \end{equation}
\end{defi}
\begin{rem}
The complexities of coarse-grained probability $p_G(O_{\n})$ in Eq.~\eqref{eq:probOrbit},  of the partition-averaged photon distribution of the $k$-th detector \eqref{eq:averPhNumberfiner}, and its coarse-grained version  \eqref{eq:averPhNumber}  are  NP-hard, as we need to sum on the number of elements in $O_n$, which may be of order $M!$.
\end{rem}

\begin{thm}\label{thm:IsoGraphs}
  The partition-averaged photon distributions introduced in Definition~\ref{def:partAvPhDistro} of two isomorphic graphs are identical up to a permutation of output modes which can be verified in polynomial time in $M$.
\end{thm}
\begin{proof}
Given  partition-averaged photon distribution of a graph $G$ we replace it by a distribution  $\lan{\tilde n_k}\ran_G, k\in[n]$, where $0\le \tilde n_1\le \cdots \le\tilde n_M$.  Clearly, this rearrangement can be done in $O(M^2)$, actually $O(M \log M)$, time.  Two isomorphic graphs will have the same rearranged partition-averaged photon distribution.
\end{proof}
\begin{lem}\label{lem:permProbDistro}
  Let $G_A$ and  $G_{\tilde A}$ be isomorphic graphs. Then the output probability distribution from GBS with encoded graphs is related by a permutation.
\end{lem}
\begin{proof}
  Consider pure events first where we present two proofs. A graph $\tilde{A}$ is isomorphic to $A$ iff there exists a permutation $\pi$ such that $\tilde{A}=P_\pi^\top AP_\pi$. Ignoring the prefactor ${c^{|\n|}}/({\n!\sqrt{\det{\s_{Q}}}})$ in Eq.~\eqref{eq:probDoubled} (it is identical for $A$ and $\tilde{A}$ -- see Lemma~\eqref{lem:determinants}), it follows that $\tilde{A}^{\op2}\ox\bbJ_{2|\n|}$ is also a permutation of $A^{\op2}\ox\bbJ_{2|\n|}$ since $\sox\equiv\ox$ for pure detection events (see Remark below Def.~\eqref{def:sox}). Hence $\haf{[\tilde{A}^{\op2}\ox\bbJ_{2|\n|}]}=\haf{[A^{\op2}\ox\bbJ_{2|\n|}]}$ and the probability expressions are invariant.

  We prove the same statement by using Eq.~\eqref{eq:ProbMixedGBS} where $C=c(A\op A),\tilde{C}=c(\tilde{A}\op\tilde{A})$ and we can ignore $c$ here by setting $c=1$. We introduce $P\df P_\pi\op P_\pi$ and write
  \begin{equation}
    \partial^{|\n|}_{\bs{\b},\ol{\bs\b}}e^{{1\over2}\bs{\g}^\top\tilde{A}^{\op2}\bs{\g}}
    =\partial^{|\n|}_{\bs{\b},\ol{\bs\b}}e^{{1\over2}(P\bs{\g})^\top A^{\op2}(P\bs{\g})}
  \end{equation}
  But that implies that the probability of a pure event remains the same since $P\bs{\b}$ by definition merely relabels the output modes and
  the partial derivatives do not care:
  \begin{equation}
    \partial^{|\n|}_{\bs{\b},\ol{\bs\b}}e^{{1\over2}(P\bs{\g})^\top A^{\op2}(P\bs{\g})}
    =\partial^{|\n|}_{\bs{\b},\ol{\bs\b}}e^{{1\over2}\bs{\g}^\top\tilde{A}^{\op2}\bs{\g}}.
  \end{equation}

  For  mixed detection events the situation is different. If one of the $n_i$'s in
  \[
  \partial^{|\n|}_{\bs{\b},\ol{\bs\b}}e^{{1\over2}(P\bs{\g})^\top A^{\op2}(P\bs{\g})}
  \]
  is different from the rest, the corresponding partial derivative breaks the symmetry and unlike the pure case one concludes that
  \begin{equation}
    \partial^{|\n|}_{\bs{\b},\ol{\bs\b}}e^{{1\over2}(P\bs{\g})^\top A^{\op2}(P\bs{\g})}
    \neq
    \partial^{|\n|}_{\bs{\b},\ol{\bs\b}}e^{{1\over2}\bs{\g}^\top{A}^{\op2}\bs{\g}}.
  \end{equation}
  However, if we permute the derivative variables (symbolically written as $\partial_{\b_i,\ol\b_i}\mapsto\partial_{(P_\pi\b_i),(P_\pi\ol\b_i)}$), we find the desired equality
  \begin{equation}\label{eq:mixedprobs1}
    \partial^{|\n|}_{P_\pi\bs{\b},P_\pi\ol{\bs\b}}e^{{1\over2}(P\bs{\g})^\top A^{\op2}(P\bs{\g})}
    =\partial^{|\n|}_{\bs{\b},\ol{\bs\b}}e^{{1\over2}\bs{\g}^\top{A}^{\op2}\bs{\g}}.
  \end{equation}
  Next, using map~\eqref{eq:betaToalphan}, we rewrite the both sides of the last equation as
  \begin{equation}\label{eq:mixedprobs2}
  \prod_{i=1}^{M}\partial_{(\hat{P}_\pi\a_i),(\hat{P}_\pi\ol\a_i)}e^{{1\over2}(\hat{P}\bs{\d})^\top(A^{\op2}\sox\bbJ_{2|\n|})(\hat{P}\bs{\d})}
    =\prod_{i=1}^{M}\partial_{\a_i,\ol\a_i}e^{{1\over2}\bs{\d}^\top(A^{\op2}\sox\bbJ_{2|\n|})\bs{\d}}.
  \end{equation}
  where $\bs{\d}\df(\bs\a,\ol{\bs\a})$ and $\hat P\df\hat{P}_\pi\oplus \hat{P}_\pi$ was introduced in Lemma~\ref{lem:hafinv}. We used Eq.~\eqref{eq:ProbMixed}, Lemma~\ref{lem:linearity} and Lemma~\ref{lem:hafinv} (the upper route in the commutative diagram to go from the LHS of~\eqref{eq:mixedprobs1} to the LHS of~\eqref{eq:mixedprobs2}). But since $(\hat{P}\bs{\d})^\top(A^{\op2}\sox\bbJ_{2|\n|})(\hat{P}\bs{\d})=\bs{\d}^\top(\hat{P}^\top(A^{\op2}\sox\bbJ_{2|\n|})\hat{P})\bs{\d}$ and  $\hat{P}$ is a permutation, the hafnian is preserved and the  output probability distribution is merely permuted.

  To conclude the proof we notice that the overall detection probability is a sum of invariant (for pure events) or permuted (for the mixed ones) probability distributions where the permutation is the same for all mixed $\n$'s.
\end{proof}
To simplify the notation in the rest of the section we write $\haf^{\,2}_{G}{(\bs{n})}\equiv\haf^{\,2}{[A\sox\bbJ_{|\bs{n}|}]}$ in Eq.~\eqref{eq:averPhNumber}. Given the stratification into orbits, it is advantageous to collect $n_k$ together with the factorial coefficients and the (squared) hafnians of a graph $G$ to $\mathsf{N}$ and $\mathsf{haf}_G$, respectively, and rewrite~\eqref{eq:averPhNumber} as
\begin{equation}\label{eq:averPhMatrix}
   \mathsf{n}_{G}={1\over\sqrt{\det{\s_{Q,G}}}}\mathsf{N}\ \mathsf{haf}_{G},
\end{equation}
where $\mathsf{n}_{G}$ is $M$-tuple of numbers.
\begin{exa}
  Let us illustrate~\eqref{eq:averPhMatrix} for a graph $G$ on $M=4$ vertices and for $|\n|=2$. There are two orbits represented by $(0,0,0,2)$ and $(0,0,1,1)$. Since the graph is  doubled, we have $M=4$ detectors and  then
    \begin{equation}\label{eq:LinEqsM4n2}
      \mathsf{n}_{G}=
      {1\over\sqrt{\det{\s_{Q,G}}}}
      \begin{bmatrix}
       2/2! & 0 & 0 & 0 & 1 & 1 & 0 & 1 & 0 & 0 \\
       0 & 2/2! & 0 & 0 & 1 & 0 & 1 & 0 & 1 & 0 \\
       0 & 0 & 2/2! & 0 & 0 & 1 & 1 & 0 & 0 & 1 \\
       0 & 0 & 0 & 2/2! & 0 & 0 & 0 & 1 & 1 & 1 \\
      \end{bmatrix}
      \begin{bmatrix}
        \haf^{\,2}_{G}{(2000)} \\[0.3em]
        \haf^{\,2}_{G}{(0200)} \\[0.3em]
        \haf^{\,2}_{G}{(0020)} \\[0.3em]
        \haf^{\,2}_{G}{(0002)} \\[0.3em]
        \haf^{\,2}_{G}{(1100)} \\[0.3em]
        \haf^{\,2}_{G}{(1010)} \\[0.3em]
        \haf^{\,2}_{G}{(0110)} \\[0.3em]
        \haf^{\,2}_{G}{(1001)} \\[0.3em]
        \haf^{\,2}_{G}{(0101)} \\[0.3em]
        \haf^{\,2}_{G}{(0011)}
      \end{bmatrix}.
    \end{equation}
    We can clearly identify the sums on the RHS of~\eqref{eq:averPhNumber}. Note that due to Corollary~\ref{cor:ofnTOTlemma} the hafnians of the $(0,0,0,2)$ orbit are zero and so are the corresponding contributions to $\mathsf{n}_{G}$.
\end{exa}
The following proof is best viewed together with the above example.
\begin{proof}[Proof of Theorem~\ref{thm:IsoGraphs}]
  Lemma~\ref{lem:permProbDistro} shows that, if graphs $G_1$ and $G_2$ are isomorphic, then the ordered set of hafnians for one graph is a permutation of the same ordered set for the other graph. This statement translates into a permutation of  $\mathsf{haf}_{G}$ introduced earlier: $\mathsf{haf}_{G_2}=\pi(\mathsf{haf}_{G_1})$. Note that the pure orbit elements are fixed points of $\pi$. Now we  observe that the $i$-th row of $\mathsf{N}$ by construction coincides with the sequence assembled from the $i$-th elements of all $\bs{n}$ (see~\eqref{eq:LinEqsM4n2}). So instead of swapping the rows of $\mathsf{N}$ we correspondingly swap these sequences in the argument of all $\haf^{\,2}_{G}$ forming $\mathsf{haf}_{G}$. But this transformation is a permutation of the set of all $\bs{n}$'s for a fixed $|\n|$ since it preserves the photon number. Hence
   \begin{equation}
     \mathsf{N}\,\pi(\mathsf{haf}_{G_1})=P_\pi(\mathsf{N})\,\mathsf{haf}_{G_1},
   \end{equation}
   where $P_\pi$ swaps the rows of $\mathsf{N}$. Since $\mathsf{N}\mathsf{haf}_{G_2}=\mathsf{N}\,\pi(\mathsf{haf}_{G_1})$   we get
    \begin{equation}
       \mathsf{N}\mathsf{haf}_{G_2}=P_\pi(\mathsf{N})\,\mathsf{haf}_{G_1}
    \end{equation}
   and thanks to Lemma~\ref{lem:determinants} we can rewrite the equality as
    \begin{equation}
       {1\over\sqrt{\det{\s_{Q,G_2}}}}\mathsf{N}\mathsf{haf}_{G_2}={1\over\sqrt{\det{\s_{Q,G_1}}}}P_\pi(\mathsf{N})\,\mathsf{haf}_{G_1}.
    \end{equation}
   But the LHS is $\mathsf{n}_{G_2}$ and the action of a permutation $P_\pi$ on the RHS implies that it is equal to $\pi(\mathsf{n}_{G_1})$. Therefore
  \[
  \mathsf{n}_{G_2}=\pi(\mathsf{n}_{G_1})
  \]
   which is the main result.

   To conclude the proof we observe that it takes only a polynomial number of steps to uncover how the partition-averaged photon distribution is permuted. We order $\mathsf{n}_{G_1}$ and $\mathsf{n}_{G_2}$ in an increasing order and if the two ordered sets differ the graphs cannot be isomorphic
\end{proof}
One could be tempted to argue that the opposite is true (that is, if the partition-averaged photon distributions the same then the graphs are isomorphic). The following counterexample shows that there is no hope for the converse of Theorem~\ref{thm:IsoGraphs}.
\begin{exa}[Counterexample based on SRG(16,6,2,2)]\label{exa:counterexa}
  SRG(16,6,2,2) is the smallest family of SRGs containing two isospectral graphs on 16 vertices. Let $|\n|=4$ which can be partitioned in five different ways. Orbits represented by $\n=(1,3)$ and $\n=(4)$ (zeros omitted) do not contribute in accord with Corollary~\ref{cor:ofnTOTlemma}. Calculating Eq.~\eqref{eq:averPhNumber} we find $\lan\lan{n_k}\ran\ran_{G_1}=\lan\lan{n_k}\ran\ran_{G_2}$. What about the less coarse-grained version, Eq.~\eqref{eq:averPhNumberfiner}. Let's check the orbit of $\n=(1,1,1,1)$ where $|O_{\n}|=1820$. Here the situations is quite interesting and generic for SRGs. The sets of hafnians differ: $\mathrm{hafs}[G_1]=(0_{992},1_{768},2_{60})$ and $\mathrm{hafs}[G_2]=(0_{984},1_{792},2_{36},3_8)$ where the subscripts count the hafnian. Yet, we find $\lan{n_k}\ran_{G_1}=\lan{n_k}\ran_{G_2}$.
\end{exa}
\begin{rem}
  Note that since the hafnian sets differ in the previous example we know that the graphs are not isomorphic. It just can't be concluded from comparing the partition-averaged photon distributions for $|\n|=4$ and it can't even be concluded from~\eqref{eq:probOrbit} since $p_{G_1}(O_{n})=p_{G_2}(O_{n})$  for all orbits  for $|\n|=4$ (including $\n=(1,1,1,1)$ again!). The first differences both in $\lan{n_k}\ran_G$ and $p_{G}(O_{n})$ appear for some orbits of $|\n|=8$. Interestingly, $\lan{n_k}\ran_G$ is always uniform for SRGs and when it differs for two nonisomorphic SRGs, it differs in a magnitude.
\end{rem}
\begin{rem}
  Similarly to the partition~\eqref{eq:probCoarsegrained}, the coarse-grained partition-averaged photon distribution~$\langle\langle{n_k}\rangle\rangle_G$ is efficiently calculable.
\end{rem}

\section{Simulations for isospectral graphs}\label{sec:simulations}

In the following section, we present the results of the GBS quantum GI algorithm applied to various SRG families and other isospectral graphs. The algorithm itself is presented in the Appendix \ref{sec:algo}. Among other graphs, we examine the SRG(35,18,9,9) family, and show that, using various detection patterns, the GBS fully distinguishes all 3854 graphs in this family. Due to the large number of photon event permutations required to calculate the probability of detection, and the classically intractable graph hafnian calculation, the results were computed in parallel using the Python Hafnian library~\cite{bjorklund2018faster} and the Titan supercomputer~\footnote{\url{https://www.olcf.ornl.gov/olcf-resources/compute-systems/titan/}}. Recall our convention for GBS encodable graphs: $C=c(A\op A)\in\bbR^{2M\times2M}$ where we set $c=1$ whenever we are allowed to.

\begin{exa}
    Let us start with the smallest connected PING in Fig.~\ref{fig:PING6v}. The hafnians of the adjacency matrices coincide so we have to look to all possible GBS-measurable submatrices. These correspond to all measurement patterns with at most one photon per mode (in this example we won't study the multiphoton contributions coming from $A\sox\bbJ_{\n}$). Hence, we can measure in total $\binom{6}{4}=15$ graphs on 4 vertices as well as 2 vertices (the subgraphs with an odd number of vertices have zero perfect matchings and therefore zero hafnian). The hafnians  of the latter (let's call them 2-hafnians) do not differ but the 4-hafnian sets do differ:
    \begin{subequations}\label{eq:hafnotation}
        \begin{align}
          {4\myhyph\!\haf{A_{G_1}}} &= (0, 0, 0, 0, 0, 0, 0, 1, 1, 1, 1, 1, 1, 1, 2)\equiv(0_7,1_7,2_1), \\
          {4\myhyph\!\haf{A_{G_2}}} &= (0_8,1_7).
        \end{align}
    \end{subequations}
\end{exa}
\begin{exa}
    Consider another example of a PING~\cite{baker1966drum} in Fig.~\ref{fig:PING9v}, this time on nine vertices.
    \begin{figure}[h]
      \resizebox{8cm}{!}{\includegraphics{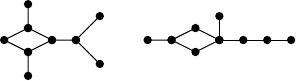}}
      \caption{ PING on nine vertices. }
      \label{fig:PING9v}
    \end{figure}
    Their hafnians are zero since the number of vertices is odd but also all (8,6,4,2)-hafnians are identical (that is, all subgraphs accessible to GBS have the same number of perfect matchings). We thus have to change the strategy and systematically investigate  multiphoton detection events by using the stratification according to the overall photon number and analyze all detection events, both pure and mixed. This is the way we will proceed in the upcoming examples. Here it turns out that the first differences between the two graphs happen for $|\n|=6$. Table~\ref{table:PING9v} on page~\pageref{table:PING9v} summarizes the result. The leftmost column contains all partitions of $|\n|=6$ (see~\eqref{eq:partition}) where each partition is represented by a naturally ordered orbit representative. The orbit size is in the second column in accordance with the discussion preceding~Eq.~\eqref{eq:probOrbit}. In experimental terms, an orbit consists of a measurement pattern and all of its permutations. The two rightmost columns contain the hafnians $\haf{[A\sox\bbJ_{|\n|}]}$. We notice a difference in three orbits (greyed): $(1,1,1,1,2),(1,1,2,2)$ and $(1,1,1,3)$ (the zeros omitted). Also note that the last four rows corresponding to events which do not occur as predicted by Lemma~\ref{lem:nTOT}. For another graph~$G_3$, isomorphic to $G_2$, we get the same hafnians for all orbits as an additional check.
    \begin{center}
    \renewcommand{\arraystretch}{1.1}
    \extrarowheight=\aboverulesep
    \aboverulesep=0pt
    \belowrulesep=0pt
    \begin{table}
           \begin{tabular}{@{}>{\columncolor{white}[0pt][\tabcolsep]}  *5c @{}}    \toprule
            Orbit representative of $O_{\n}$ & $|O_{\n}|$  & $\mathrm{hafs}[G_1]$ & $\mathrm{hafs}[G_2]$ \\
            \midrule
             $(0,0,0,1,1,1,1,1,1)$  & $\binom{9}{6}$              & $(0_{69}, 1_{13},2_{2})$  & $(0_{69}, 1_{13},2_{2})$ \\
            \rowcolor{lightgray}
             $(0,0,0,0,1,1,1,1,2)$  & $\binom{9}{5}\binom{5}{1}$  & $(0_{586}, 2_{41},4_{3})$  & $(0_{585}, 2_{42},4_{3})$\\\addlinespace[0.1em]
             \rowcolor{lightgray}
             $(0,0,0,0,0,1,1,2,2)$  & $\binom{9}{4}\binom{4}{2}$  & $(0_{698}, 2_{42},4_{12},6_4)$  & $(0_{700},2_{42},4_{10},6_4)$\\\addlinespace[0.1em]
             $(0,0,0,0,0,0,2,2,2)$  & $\binom{9}{3}$              & $(0_{84})$  &  $(0_{84})$ \\
             \rowcolor{lightgray}
             $(0,0,0,0,0,1,1,1,3)$  & $\binom{9}{4}\binom{4}{1}$  & $(0_{500},6_4)$  &   $(0_{499},6_5)$\\\addlinespace[0.1em]
             $(0,0,0,0,0,0,1,2,3)$  & $\binom{9}{3}3!$            & $(0_{478},6_{26})$  &   $(0_{478},6_{26})$ \\\addlinespace[0.1em]
             $(0,0,0,0,0,0,0,3,3)$  & $\binom{9}{2}$              & $(0_{27},6_9)$  &   $(0_{27},6_9)$\\
             $(0,0,0,0,0,0,1,1,4)$  & $\binom{9}{3}\binom{3}{1}$  & $(0_{252})$  &  $(0_{252})$ \\\addlinespace[0.1em]
             $(0,0,0,0,0,0,0,2,4)$  & $\binom{9}{2}\binom{2}{1}$  & $(0_{72})$  &  $(0_{72})$ \\
             $(0,0,0,0,0,0,0,1,5)$  & $\binom{9}{2}\binom{2}{1}$  & $(0_{72})$  &  $(0_{72})$ \\
             $(0,0,0,0,0,0,0,0,6)$  & 9                           & $(0_{9})$  &  $(0_{9})$ \\
             \bottomrule
             \hline
            \end{tabular}\\ \vskip .3cm
    \caption{All measurement patterns and their permutations~$O_{\n}$ (orbits) for the PING on nine vertices in Fig.~\ref{fig:PING9v} for $|\n|=6$. The second column is the orbit size and in the last two columns we list the hafnians whose total number (the sum of subscripts) equals $|O_{\n}|$. Note that the notation we are using for the hafnian sets is defined in~\eqref{eq:hafnotation}.}
    \label{table:PING9v}
    \end{table}
    \end{center}

    Another interesting piece of information is the actual partition-averaged photon distribution  for a given orbit provided by~\eqref{eq:averPhNumberfiner} or~\eqref{eq:averPhNumber} (we omit the determinants in this example). For non-SRGs the partition-averaged photon distribution is typically non-flat. To illustrate~\eqref{eq:averPhNumberfiner} let's choose an orbit  where no difference was found: $O_{\n}$ for $\n=(1,2,3)$.  The first two plots in Fig.~\ref{fig:v9distroG1G2G3} are clearly different (that is, non-permutationally invariant). In accordance with the result of Section~\ref{subsec:meanphoton}, namely Theorem~\ref{thm:IsoGraphs}, this is enough to decide that the two graphs are not isomorphic. For a graph $G_3$ isomorphic to $G_2$ we notice a mere permutation of bars in the rightmost panel of Fig.~\ref{fig:v9distroG1G2G3} again in accordance with Theorem~\ref{thm:IsoGraphs}.

    A similar conclusion follows from the analysis of~\eqref{eq:averPhNumber} and the situation is depicted in Fig.~\ref{fig:v9distron8G1G2G3} for~$|\n|=8$.
    \begin{figure}[t]
      \resizebox{9.3cm}{!}{\includegraphics{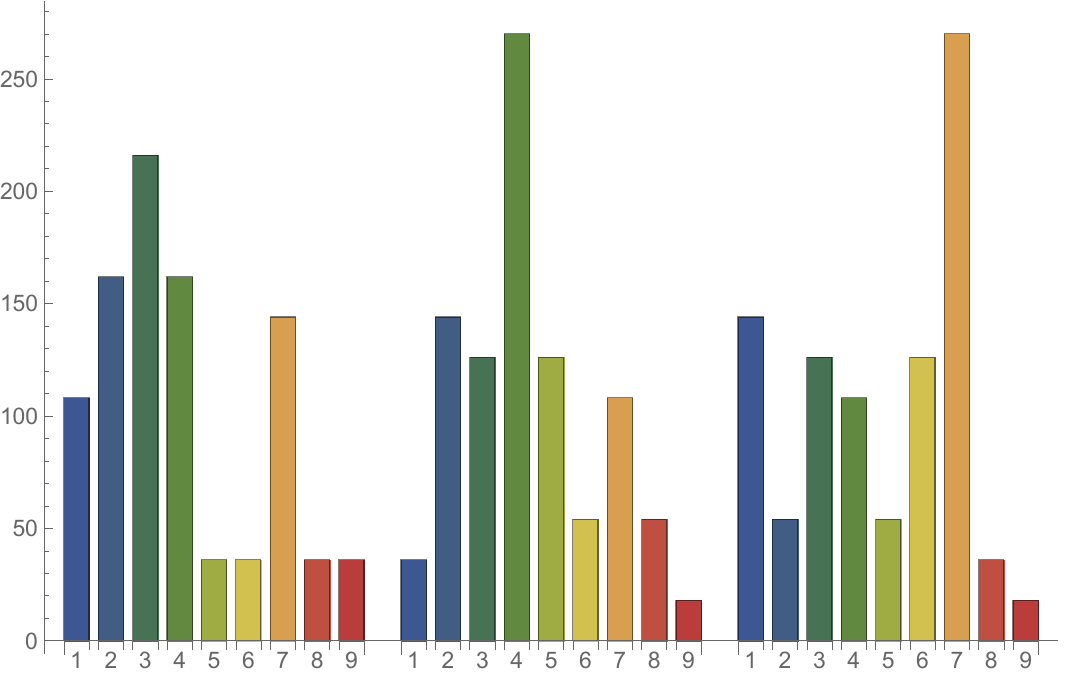}}
      \caption{Partition-averaged photon distribution, Eq.~\ref{eq:averPhNumberfiner}, for $G_1,G_2$ and  $G_3\simeq G_2$ for the orbit of $\n=(1,2,3)$. The $x$ axis labels the detectors. Note that since we omitted the determinantal prefactor  and set $c=1$ the distribution is not normalized.}
      \label{fig:v9distroG1G2G3}
    \end{figure}
    \begin{figure}[t]
      \resizebox{9.3cm}{!}{\includegraphics{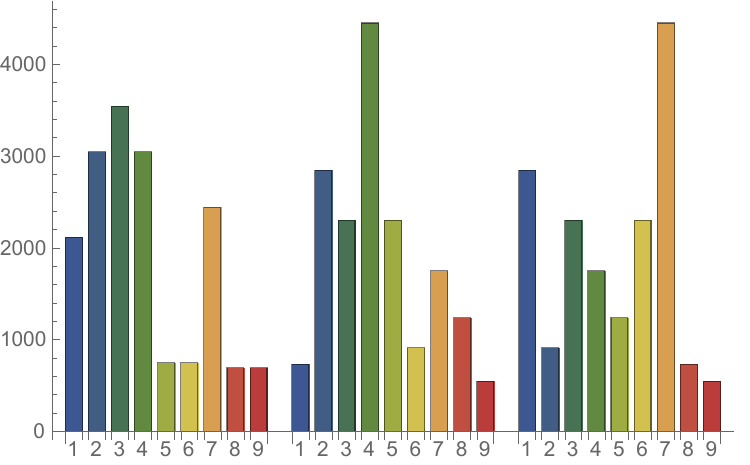}}
      \caption{Coarse-grained partition-averaged photon distribution, Eq.~\ref{eq:averPhNumber}, for $G_1,G_2$ and  $G_3\simeq G_2$ for all orbits contributing to $|\n|=8$. The $x$ axis labels the detectors. Note that since we omitted the determinantal prefactor  and set $c=1$ the distribution is not normalized.}
      \label{fig:v9distron8G1G2G3}
    \end{figure}

\end{exa}
\begin{exa}
  Regular isospectral nonisomorphic graphs appear to be somewhere between PINGs and SRGs in terms of the difficulty to distinguish them. We analyzed a pair of graphs on ten vertices introduced in~\cite[page~110, (a) and (b)]{little2006combinatorial}. Neither the coarse-grained photon distribution, Eq.~\eqref{eq:averPhNumber}, nor its probability equivalent differ for the two graphs for any tested orbit of $\n$. This is what we witnessed for all examined SRGs as well. But there is a difference, most likely related to the fact that regular graphs have less symmetry than SRGs. First, a difference in $\lan{n_k}\ran_G$ but not in $p_G(O_{\n})$ appears for some orbits of $|\n|=6$. For $|\n|=8$ both quantities differ in an ever increasing number orbits. What makes regular graphs different from SRGs is that $\lan{n_k}\ran_G$ is not uniform (c.f. with the example and remark at the end of Section~\ref{subsec:meanphoton}). This is more similar to the PINGs we mentioned previously.  So we may get some information on the actual permutation operation from $\lan{n_k}\ran_{G_{1,2}}$  if we cannot find any difference for any $\n$.
\end{exa}
      \begin{figure}[b]
      \resizebox{14cm}{!}{\includegraphics{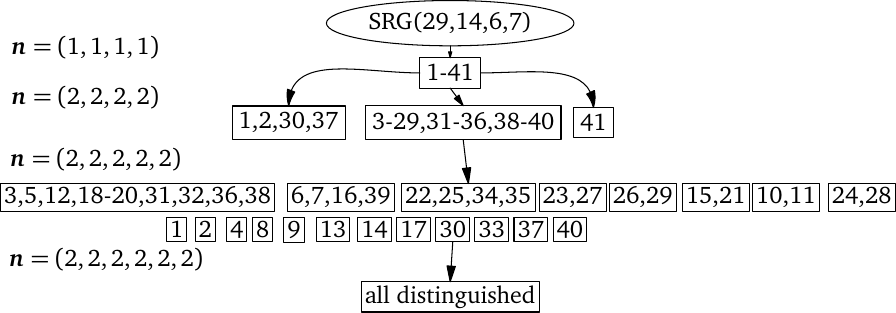}}
      \caption{The SRG(29,14,6,7) family of 41 isospectral graphs can be fully distinguished by considering orbit $(2,2,2,2,2,2)$. The numbers in the rectangular boxes label the graphs according to~\cite{spence2018}. Considering `smaller' orbits (in terms of $|\n|$ or the number of nonzero $n_i$'s) typically leads to a partial separation. Orbit $(2,2,2,2,2,2)$ may not be the smallest one to distinguish all the graphs.}
      \label{fig:SRG29}
    \end{figure}
\begin{exa}[SRG(29,14,6,7)]
    Fig.~\ref{fig:SRG29} summarizes the path to distinguish all 41 isospectral graphs. As the starting point we took orbit $O_{\n}$ for $\n=(1,1,1,1)$. This orbit has no distinguishing power. This is indicated by a single square box containing all graphs. One could start with a measurement event containing more single photons but the problem is, as the number of vertices increases, the orbit size grows rapidly making the simulations rather resource-expensive. Also, it is desirable to find the orbit with the smallest possible $|\n|$ distinguishing all graphs to (i) heuristically assess the performance of our algorithm and (ii) make sure that the physical resources needed to run the algorithm are not excessive. This is because the smaller $|\n|$ is the less squeezing we need in an actual experiment. Sometimes, however, a smaller $|\n|$ does not guarantee a faster simulation. In the current example the orbit of $\n=(2,2,2,2,2,2)$ where $|\n|=12$ is much more computationally feasible than $\n=(1,2,3,4)$ where $|\n|=10$. This is because $|O_{\n}|=475020$ for the former but $|O_{\n}|=11400480$ for the latter. So perhaps the orbit of $\n=(1,2,3,4)$ distinguishes them all. Hence without exploring all alternative routes we cannot claim optimality (here and in any other example).
\end{exa}
      \begin{figure}[t]
      \resizebox{14cm}{!}{\includegraphics{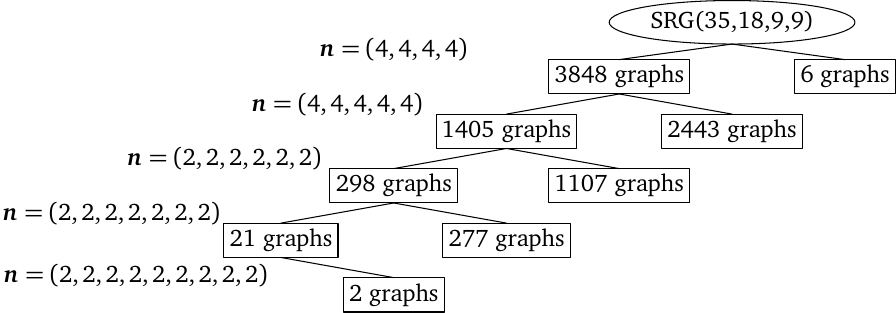}}
      \caption{The SRG(35,18,9,9) family of 3854 isospectral graphs can be fully distinguished by starting with the orbit of $\n=(4,4,4,4)$ up to $\n=(2,2,2,2,2,2,2,2,2)$.}
      \label{fig:SRG35}
    \end{figure}
\begin{exa}[SRG(35,18,9,9)]
    The analysis of the biggest family of SRGs we studied is summarized in Fig.~\ref{fig:SRG35}. Starting from the top, the orbit representatives on the left indicate how successful they are in distinguishing the graphs. The right box contains the number of \emph{newly} distinguished graphs whereas the left box contains the number of remaining graphs. The effect is cumulative so for example the orbit of $\n=(4,4,4,4,4)$ together with $\n=(4,4,4,4)$ distinguishes 2443 graphs. Our computational resources were not enough to distinguish the two remaining graphs using Theorem~\ref{thm:mainthm}. The necessary condition developed in Sec.~\ref{subsec:hierarchyTower} was used instead.
\end{exa}
\begin{exa}[SRG(16,6,2,2)]  The two graphs can be easily distinguished by our method but in this case we illustrate the probability function for all partitions and their orbits up to $|\n|=14$. In Fig.~\ref{fig:SRG16v} we plot Eq.~\ref{eq:probOrbit} for orbits whose probability is nonzero (so their number is less than given by partitioning $|\n|$). The $x$ axis labels these orbits and in the plot we indicate the actual partitioning by white and gray background. Even for a fixed $|\n|$ some orbits are more likely than others. We observed that the probability is correlated to the size of the orbit. This confirms our intuition from the paragraph before~Eq.~\eqref{eq:probOrbit}.
  \begin{figure}[h]
  \resizebox{13cm}{!}{\includegraphics{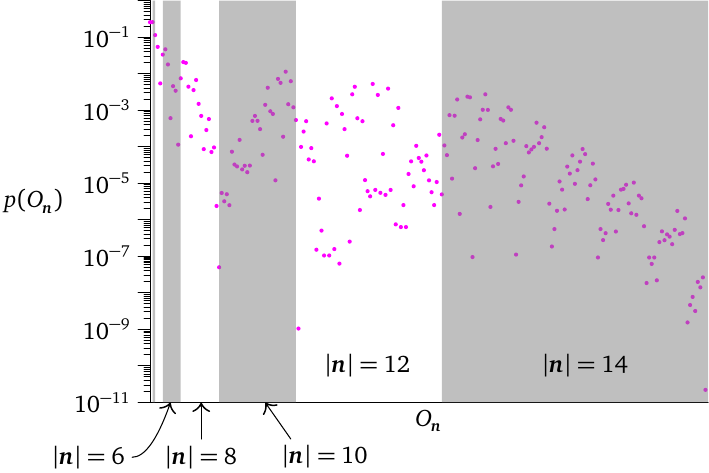}}
  \caption{Probabilities of various orbits. Each point is a probability of an orbit  $O_{\n}$ given by Eq.~\eqref{eq:probOrbit} for one of the graphs from the SRG(16,6,2,2) family. The prefactor in~\eqref{eq:probOrbit} is $1/\sqrt{\det{\s_{Q,G}}}=((-1+4c^2)^{15}(-1+36c^2))^{1/2}$ where we set c=1/6.9. Orbits whose probability is zero are omitted.}
  \label{fig:SRG16v}
\end{figure}
\end{exa}

\section{Open problems and Discussion}\label{sec:discussion}

\subsection{Open questions}

\begin{enumerate}[label=\bfseries{Q\arabic*.}]
 \item   Can we replace in Theorem \ref{thm:mainthm} the quantity
    $\sum_{\sigma\in \mathfrak{S}_n} \sqrt{p(\bs{n}_{\sigma})}$ with $\sum_{\sigma\in \mathfrak{S}_n} p(\bs{n}_{\sigma})$?
  \item The main result of this paper is a necessary and sufficient condition for two isospectral graphs to be isomorphic. We found a complete set of graph invariants. However, this is only half satisfactory because we don't know where the difference between two graphs `kicks in'. Without this knowledge we can use the iff condition only in one direction. The ideal situation would be to have a deterministic or probabilistic criterion for the existence of such a \emph{threshold} orbit as a polynomial function of the graph size. The numerical experiments are favorable as far as the polynomial growth goes -- there is no indication that the threshold value grows fast. As the SRG example at the end of Section~\ref{subsec:meanphoton} shows, this is actually a more subtle problem: the set of hafnians may be different which is one sufficient condition as shown in Section~\ref{subsec:graphinv} but their sum of squares is forming $p_{G_i}(O_{\n})$ (the coarse-grained probability as another sufficient condition) is the same. The latter often comes `later', that is, for higher orbits than the former.
  \item When a threshold orbit $\n$ is reached we made the following  observation: upon examining orbits $\m$ for $|\m|>|\n|$ satisfying $\m>\n$ ($m_i\geq n_i,\forall i$ and at least one $n_i\neq m_i$) the set of hafnians is again different. We call this a \emph{branching effect}. It seems to be a probabilistic effect but from our experiments it holds overwhelmingly. This results in a dramatic increase of orbits with different hafnian sets as we increase $|\m|$ and greatly helps in the practical use of our algorithm. The intuition behind this behavior is that when the hafnian sets differ for an orbit of $\n$, then another orbit of $\m>\n$ (obtained by adding two new rows/columns or copies thereof to the adjacency matrix corresponding to $\n$) contains as its subgraphs all the graphs that already had different hafnian sets. The question to answer is how likely this effect is to occur.
  \item A problem closely related to the previous question is how long it takes to approach a true probability distribution within a given precision for a chosen orbit of interest. This is typically answered by the methods that are standard in random graph and probability theory.
  \item We did not address two related experimental points in this paper.  The first one is the photon losses, which will make the statistics of two isomorphic graphs different.  Hence one needs to come up with cutoff estimates for distinguishing nonisomorphic graphs. The second one is how easy to estimate probabilities of photon counting experiments.
\end{enumerate}

\subsection{Scalability discussion}

As discussed in the previous subsection, we don't know at what point two non-isomorphic graphs start to differ when sampled using a GBS device. Presumably this critical $|\n|$ grows with the graph size and the numerical evidence suggests that the growth is not exponential or even fast in general in the graph size. Nonetheless, even a moderate rate of growth of the threshold orbit $\n$ with the graph size $M$ could be fatal for large graphs. This is because the physical interpretation of $0<c<\|A\|_2$ that directs the amount of squeezing -- a quantity directly related to $|\n|$. As we see from Fig.~\ref{fig:SRG16v}, the class with a (nearly) maximal probability  (a single orbit in fact) is $\n=(0)_M$ and this is a generic case. We can tune the~$c$ parameter by increasing squeezing such that the probability of a different (desired) orbit $O_{\n}$ increases. This typically makes the vacuum contribution smaller but still dominating the probability landscape. What happens is that the desired orbit's probability $p_{G_A}(O_{\n})$ increases with respect to the vacuum and other lower contributions but the probability flattens and therefore its values are inevitably impractically low to gather enough statistics reasonably fast. To put it differently, there is no significant peak (or a concentration effect) for the desired orbit $O_{\n}$. The situation is a bit alleviated by a heuristic observation that beyond the threshold orbit $O_{\n}$, where the difference in the hafnian set is first observed, the orbits $O_{\m}$ where $\n<\m$ overwhelmingly add to distinguishability. But it is tempting to avoid this `probability dilution' altogether. The rise of the number of orbits for a fixed $|\n|$ is equivalent to the integer partition problem briefly discussed towards the end of Section~\ref{subsec:graphinv}.

Let us look closer at this problem. Any graph $G_A$ can be encoded as a pure covariance matrix whose circuit decomposition consists of an array of $M$ squeezers $S$ followed by an interferometer $U$ on $M$ modes. Since the eigenvalues (more precisely the singular values) of $A$ are related to  the squeezing parameters, two isospectral graphs, $G_A,G_{\tilde A}$, have the same $S$ -- the fact already used in Lemma~\ref{lem:probCoarseGr}. The input to the interferometer is given by~\eqref{eq:Sproduct}. Suppose that we know or suspect that a given $|\n|$ contains orbits that are different if two graphs are not isomorphic. Since an interferometer is a passive unitary operator we know the input state responsible for it -- it is the state whose coefficients are $\a_{i\n}(r_1,\dots,r_M)$ from~\eqref{eq:Sproduct}. So the task becomes to prepare such a state. This could be a computationally hard task. Even though $\bigotimes_{k=1}^MS(r_k)\ket{0}$ is factorized the states $\sum_{i=1}^{\binom{|\n|/2+M-1}{|\n|/2}}\a_{i\n}(r_1,\dots,r_M)\ket{(\n,M)}_i$ living in the completely symmetric subspaces  are, in general, entangled. There are two problems, though. First, generation of such states could potentially require a circuit of a great depth. Second, even if $\n$ scales favorably with $M$,  the number of coefficients $\a_{i\n}(r_1,\dots,r_M)$ can be overwhelming to work with for $M\gg0$ and so we may run into the issue of  tractability to describe the necessary unitary operation. If these issues were resolved we would gained a complete control over the output distribution behavior. But we would also switch from GBS to a generalized (multiphoton) boson sampling (BS)~\cite{aaronson2011computational}. The orbit representative is a state with a given photon number per input mode. The difference compared to BS is that we do not require the input/output state to be in the 0,1 subspace per mode. Hence we arrived (by a detour) to the output probability function whose form is most likely governed by some permanent function~\cite{scheel2004permanents} -- a form which will most likely be a variation on the reduced Kronecker product we have introduced in Sec.~\ref{subsec:suppresults}.

The use of a fixed input photon distribution will dramatically change the odds of detecting the difference in the probability distribution. We can take a look at Fig.~\ref{fig:SRG16v} for, say, $|\n|=14$ and $p_{G_A}(O_{\n})$ for all orbits in this partition will be considerably higher with their mutual ratio preserved. We leave for a future exploration the question if the probability distribution is always skewed (or even concentrated) such that some events (preferably the ones where there is a difference)  are overwhelmingly likely than the others. What can we do if this is not the case? The thing we cannot do is to coarse-grain the probability more than in Eq.~\eqref{eq:probOrbit} due to lemma~\ref{lem:probCoarseGr}.
Then, one option would be to coarse-grain the probability more than in~\eqref{eq:probOrbit} but less than in~\eqref{eq:probCoarsegrained}. The reason for this effort is to have a favorable scaling. Recall that the number of partitions of $|\n|$ grows exponentially with $|\n|$. If we partially coarse-grain a given orbit we may keep the scaling polynomial and still detect a difference for non-isomorphic graphs. The question is how to split a given orbit. At this point we can offer only certain heuristic rules based on our simulations.

\section*{acknowledgments}

The authors greatly appreciate thorough reading and suggested modifications of the manuscript by Robert Israel. This  research  used  resources  of  the Oak  Ridge  Leadership  Computing  Facility  at  the  Oak Ridge  National  Laboratory,  which  is  supported  by  the Office of Science of the U.S. Department of Energy under Contract No.~DE-AC05-00OR22725.

\bibliographystyle{unsrt}


\appendix
\section{Basic hardware setup}\label{sec:experiment}

The basis for our graph isomorphism method is a near-term photonic quantum processor, specifically a GBS apparatus. This apparatus consists of three main components. First, squeezed states are generated in $M$ quantum-optical modes. These states are then sent through an $M$-port linear-optical interferometer. Finally, a photon-number-resolving measurement is performed on each of the $M$ output modes. The first two steps lead to the preparation of a zero-mean quantum-optical Gaussian state, which can be described efficiently using a covariance matrix. For a single mode, the covariance matrix has dimension $2\times 2$, and encodes the covariances of the canonical quadrature operators ($\hat{x}$, $\hat{p}$) of that mode:
\begin{equation}\label{eq:covmatrix}
 \sigma_{ij} = \tfrac{1}{2}\langle \hat{\xi}_i \hat{\xi}_j + \hat{\xi}_j \hat{\xi}_i\rangle - \langle \hat{\xi_i} \rangle\langle \hat{\xi_j} \rangle,
\end{equation}
with $\hat{\xi}_k\in\{\hat{x}_k, \hat{p}_k\}$. For $M$ modes, we have $M$ pairs of quadrature operators and an $2M\times 2M$ covariance matrix built from the set $\bs{\hat{\xi}}\in\{\hat{x}_1,\hat{p}_1,\dots, \hat{x}_M,\hat{p}_M\}$.

Multimode Gaussian states themselves are of limited interest in quantum computing. While they can be prepared with quantum hardware and exhibit entanglement, the covariance matrix scales linearly in the number of modes, so they can be efficiently simulated classically. However, when we introduce the photon-number measurement, the story changes. A single photon-number measurement in mode $k$ will return a nonnegative integer $n_k\in\bbN^+$, representing the number of photons which were detected. For measurement on $M$ modes, we denote the collective photon-number output pattern by $\n=(n_1,\dots,n_M)$ and call it a detection event. From~\cite{hamilton2016gaussian}, whenever $n_i=1,\forall i$ the probability of this detection event is proportional to a function called the \emph{hafnian}~\cite{caianiello1953quantum} (see Def.~\ref{def:haf}):
\begin{equation}\label{eq:p11111}
 p(1,\dots,1)={1\over\sqrt{\det{\s_Q}}}\haf{C},
\end{equation}
where the matrix $C$ is obtained from  $\s_Q$ by basic matrix transformations (see Eq.~\eqref{eq:sigmaQ}).

Unlike the simulation of Gaussian states, computing the hafnian is a \#P-hard problem. In addition, approximating the GBS photon-number distribution is believed to be computationally hard~\cite{hamilton2016gaussian}. Thus, by combining Gaussian states with photon-number measurements -- representing the wavelike and particle-like properties of light, respectively -- we have a physical sampling apparatus whose behaviour is classically hard to replicate.  This paper explores the question of how we can leverage this GBS device for the graph isomorphism problem, specifically, how we set the squeezing and interferometer parameters to represent the problem, and how to interpret the photon-number measurement outcomes to solve the problem.

\newpage
\section{Algorithm}\label{sec:algo}
\begin{algorithm}
    \caption{GBS graph isomorphism algorithm: returns the GI invariant of adjacency matrix $A$ with respect to orbit $o$, considering hafnians' to the $n$th power.}
    \begin{algorithmic}[1] 
        \Function{GBS\_cert}{$A$, $o$, $n$}
            \LineComment{\textit{Generate a list of unique permutations of the orbit}}
            \State perms $\gets$ \text{unique\_permutations}(o)
            \State result $\gets$ array[len(perms)]
            \For{p $\in$ perms}
	            \LineComment{\textit{Generate the reduced Kronecker product of matrix $A$}}
	            \State Ap $\gets$ \text{kron\_reduced}(A, p)
	            \LineComment{\textit{Append the hafnian to the $n$th power}}
	            \State result $\gets$ append(result, haf(Ap)\textsuperscript{n})
	        \EndFor\\
	        \LineComment{\textit{Calculate the sum of hafnians}}
            \State hafSum $\gets$ sum(results)\\
	        \LineComment{\textit{Calculate the mean photon distribution for the orbit}}
	        \State phDist $\gets $ array[len(o)]
	        \For{i $\in$ len(perms)}
	            \For{j $\in$ len(o)}
	               \State phDist[j] $\gets$ phDist[j] + perms[i, j]*result[i]
	            \EndFor
	        \EndFor
        \EndFunction
    \end{algorithmic}
\end{algorithm}

\begin{algorithm}
    \caption{Function to return the reduced Kronecker product of matrix $A$, given a sequence of integers $n$ indicating the multi-mode photon detection event.}
    \begin{algorithmic}[1] 
        \Function{kron\_reduced}{A, n}
            \State rows $\gets$ array()
            \For{i $\in$ len(n)}
                \For{j = 0, n[i]}
                    \State rows $\gets$ append(rows, i)
                \EndFor
            \EndFor
            \State \textbf{return} A[rows][rows]
        \EndFunction
    \end{algorithmic}
\end{algorithm}

\begin{algorithm}
    \caption{Function to return unique permutations of an orbit}
    \begin{algorithmic}[1] 
        \Function{unique\_permutations}{orbit}
            \If{len(orbit) = 1}
                \LineComment{\textit{If orbit is length 1, return the value}}
                \State \textbf{return} orbit[0]
            \Else
                \LineComment{\textit{Else, store the list of unique elements in the orbit}}
                \State elements $\gets$ \text{drop\_duplicates}(orbit)
                \For{e $\in$ elements}
                    \LineComment{\textit{Unique elements except $e$}}
                    \State remaining $\gets$ elements - e
                    \For{p $\in$ unique\_permutations(remaining)}
                        \LineComment{\textit{Use recurrence to concatenate element $e$ with all remaining permutations of elements}}
                        \State \textbf{return} e + p
                    \EndFor
                \EndFor
            \EndIf
        \EndFunction
    \end{algorithmic}
\end{algorithm}

\mbox{\ }

\newpage
\section{List of symbols}\label{sec:symbols}
In this appendix we summarize with a table the most important symbols and their meaning.
\begin{center}
\begin{tabular}{ l | c | r }
\bf{description} & \bf{notation} & \bf{defined at page} \\
\hline
vectors &$\n=(n_1,\dots,n_M),\x=(x_1,\dots,x_M)$&page 3\\
partial derivative operator    &  $\partial^{n_i}_{x_i,\ol x_i}={\partial^{n_i}\over\partial x_i^{n_i}}{\partial^{n_i}\over\partial\ol{x}_i^{n_i}}$            & page 4\\
symmetric group of bijections $\sigma:[n]\to[n]$&$\mathfrak{S}_n$ & page 4\\
$N\times N$ hermitian, psd and positive definite &$\rH_{N}\supset \rH_{+,N}\supset \rH_{++,N}$& page 4\\
hafnian &$ \haf{C}$& page 4\\
measurement probability of $\n$&$p(\n)$& page 4\\
covariance matrix & $\s_Q$ & page 5\\
graph &$G$& page 5\\
$i-th$ eigenvalue of $G$ &$\lambda_i(G)$& page 5\\
strongly regular graph with parameters&SRG$(N,k,\lambda,\mu)$& page 6\\
all-ones $|\n|\times |\n|$ matrix &$\bbJ_{|\n|}$& page 7\\
reduced Kronecker product &$A\sox\bbJ_{|\n|}$& page 7\\
moments of Gaussian distribution induced by $\Sigma$&$\mu_{n_1,\dots,n_{2M}}(\Sigma)$& page 11\\
multinomial coefficients &${\n\choose \ba_1,\ba_2,\ldots,\ba_m}$& page 15\\
symmetrized sums &$\mu(\n,A)$& page 16\\
probability of orbit $O_{\n}$&$p_G(O_{\n})$& page 18\\
coarse-grain probability distribution& $p_G(|\n|)$& page 19\\
 $k$-mode partition-averaged photon distribution &$\lan{n_k}\ran_G$&  page 20\\
coarse-grained version &$ \langle\langle{n_k}\rangle\rangle_G$& page 20\\
\end{tabular}
\end{center}

\end{document}